\documentclass[11pt,journal,twoside,onecolumn,draftcls]{IEEEtran}
\date{}

\usepackage{graphicx,subfigure}
\usepackage{amsfonts,amsmath,amssymb, mathrsfs}
\usepackage{epic,eepic,eepicemu}
\usepackage{epsf}
\usepackage{epsfig}
\usepackage{graphics}
\usepackage{psfrag}

\newtheorem{theorem}{Theorem}
\newtheorem{definition}{Definition}
\newtheorem{corollary}{Corollary}

\newtheorem{lemma}{Lemma}

\makeatletter
\def\@cite#1#2{[\if@tempswa #2 \fi #1]}
\makeatother

\newcommand{\beq}{\begin{equation}}
\newcommand{\eeq}{\end{equation}}
\newcommand{\bea}{\begin{eqnarray}}
\newcommand{\eea}{\end{eqnarray}}

\long\def\symbolfootnote[#1]#2{\begingroup
\def\thefootnote{\fnsymbol{footnote}}\footnote[#1]{#2}\endgroup}

\long\def\comment#1{}

\begin{document}
\title{\huge{Evaluation of Marton's Inner Bound for the General Broadcast Channel}}
\author{Amin Aminzadeh Gohari and Venkat Anantharam \\
Department of Electrical Engineering and Computer Science \\ University of California, Berkeley\\
\texttt{\small $\{$aminzade,ananth$\}$@eecs.berkeley.edu} \\
}\maketitle
\begin{abstract}
The best known inner bound on the two-receiver general broadcast
channel without a common message is due to Marton \cite{Marton}.
This result was subsequently generalized in \cite[p. 391, Problem
10(c)]{CsiszárKörner} and \cite{GelfandPinsker} to broadcast channels with a common message. However the latter region is not
computable (except in certain special cases) as no bounds on the
cardinality of its auxiliary random variables exist. Nor is it even
clear that the inner bound is a closed set. The main obstacle in
proving cardinality bounds is the fact that the traditional use of the Carath\'{e}odory
theorem, the main known tool for proving cardinality bounds, does
not yield a finite cardinality result. One of the main contributions of this paper is the introduction of 
a new tool based on an
identity that relates the second derivative of the Shannon entropy
of a discrete random variable (under a certain perturbation) to the
corresponding Fisher information. In order to go beyond the
traditional Carath\'{e}odory type arguments, we identify certain
properties that the auxiliary random variables corresponding to the
extreme points of the inner bound need to satisfy. These properties are then
used to establish cardinality bounds on the auxiliary random
variables of the inner bound, thereby proving the computability of
the region, and its closedness.

Lastly, we establish a conjecture of \cite{NairZizhou} that Marton's inner bound and the recent outer
bound of Nair and El Gamal do not match in general.
\end{abstract}
\section{Introduction}
In this paper, we consider two-receiver general broadcast channels.
A two-receiver broadcast channel is characterized by the conditional
distribution $q(y,z|x)$ where $X$ is the input to the channel and
$Y$ and $Z$ are the outputs of the channel at the two receivers. Let
$\mathcal{X}$, $\mathcal{Y}$ and $\mathcal{Z}$ denote the alphabet
set of $X$, $Y$ and $Z$ respectively. The transmitter wants to send
a common message, $M_0$, to both the receivers and two private
messages $M_1$ and $M_2$ to $Y$ and $Z$ respectively. Assume that
$M_0$, $M_1$ and $M_2$ are mutually independent, and $M_i$ (for
$i=0,1,2$) is a uniform random variable over set $\mathcal{M}_i$.
The transmitter maps the messages into a codeword of length $n$
using an encoding function $\zeta:
\mathcal{M}_0\times\mathcal{M}_1\times\mathcal{M}_2\rightarrow
\mathcal{X}^n$, and sends it over the broadcast channel $q(y,z|x)$
in $n$ times steps. The receivers use the decoding functions
$\vartheta_y:\mathcal{Y}^n \rightarrow
\mathcal{M}_0\times\mathcal{M}_1$ and $\vartheta_z:\mathcal{Z}^n
\rightarrow \mathcal{M}_0\times\mathcal{M}_2$ to map their received
signals to $(\widehat{M_0}^{(1)}, \widehat{M_1})$ and
$(\widehat{M_0}^{(2)}, \widehat{M_2})$ respectively. The average
probability of error is then taken to be the probability that
$(\widehat{M_0}^{(1)}, \widehat{M_1}, \widehat{M_0}^{(2)},
\widehat{M_2})$ is not equal to $(M_0, M_1, M_0, M_2)$.

The capacity region of the broadcast channel is defined as the set
of all triples $(R_0, R_1, R_2)$ such that for any $\epsilon>0$,
there is some integer $n$, uniform random variables $M_0$, $M_1$,
$M_2$ with alphabet sets $|\mathcal{M}_i|\geq 2^{n(R_i-\epsilon)}\ $
(for $i=0,1,2$), encoding function $\zeta$, and decoding functions
$\vartheta_y$ and $\vartheta_z$ such that the average probability of
error is less than or equal to $\epsilon$.

The capacity region of the broadcast channel is not known except in
certain special cases. The best achievable region of triples $(0,
R_1, R_2)$ for the broadcast channel is due to Marton \cite[Theorem
2]{Marton}. Marton's work was subsequently generalized in \cite[p.
391, Problem 10(c)]{CsiszárKörner}, and Gelfand and Pinsker
\cite{GelfandPinsker} who established the achievability of the
region formed by taking union over random variables $U,V,W,X,Y,Z$,
having the joint distribution $p(u,v,w,x,y,z)=p(u,v,w,x)q(y,z|x)$,
of
  \begin{eqnarray}
    R_0, R_1, R_2&\geq& 0;\nonumber\\
    R_0&\leq& \min(I(W;Y),I(W;Z));\label{eqn:DefOfMKG1}\\
    R_0+R_1&\leq& I(UW;Y);\label{eqn:DefOfMKG2}\\
    R_0+R_2&\leq& I(VW;Z);\label{eqn:DefOfMKG3}\\
    R_0+R_1+R_2&\leq&
I(U;Y|W)+I(V;Z|W)-I(U;V|W)\nonumber\\&&+\min(I(W;Y),I(W;Z)).\label{eqn:DefOfMKG4}
  \end{eqnarray}
In Marton's original work, the auxiliary random variables $U, V$ and
$W$ are finite random variables. We however allow the auxiliary
random variables $U, V$ and $W$ to be discrete or continuous random
variables to get an apparently larger region. The main result of
this paper however implies that this relaxation will not make the
region grow. We refer to this region as Marton's inner bound for
the general broadcast channel. Recently Liang and Kramer reported an
apparently larger inner bound to the broadcast channel
\cite{LiangKramerInner}, which however turns out to be equivalent to
Marton's inner bound \cite{LiangKramerPoor}. Marton's inner bound
therefore remains the currently best known inner bound on the
general broadcast channel. Liang, Kramer and Poor showed that in
order to evaluate Marton's inner bound, it suffices to search over
$p(u,v,w,x)$ for which either $I(W;Y)=I(W;Z)$, or
$I(W;Y)>I(W;Z)$\&$V=constant$, or $I(W;Y)<I(W;Z)$\&$U=constant$
holds \cite{LiangKramerPoor}. This restriction however does not lead
to a computable characterization of the region.

Unfortunately Marton's inner bound is not computable (except in
certain special cases) as no bounds on the cardinality of its
auxiliary random variables exist. A prior work by Hajek and Pursley
derives cardinality bounds for an earlier inner bound of Cover and
van der Meulen for the special case of $X$ is binary, and $R_0=0$
\cite{HajekPursley}; Hajek and Pursley showed that $X$ can be taken
as a deterministic function of the auxiliary random variables
involved, and conjectured certain cardinality bounds on the
auxiliary random variables when $|\mathcal{X}|$ is arbitrary but
$R_0$ is equal to zero. For the case of non-zero $R_0$, Hajek and
Pursley commented that finding cardinality bounds appears to be
considerably more difficult. The inner bound of Cover and van der
Meulen was however later improved by Marton. A Carath\'{e}odory-type
argument results in a cardinality bound of
$|\mathcal{V}||\mathcal{X}|+1$ on $|\mathcal{U}|$, and a cardinality
bound of $|\mathcal{U}||\mathcal{X}|+1$ on $|\mathcal{V}|$ for
Marton's inner bound. This does not lead to fixed cardinality bounds
on the auxiliary random variables $U$ and $V$. The main result of
this paper is to prove that the subset of Marton's inner bound
defined by imposing extra constraints $|\mathcal{U}|\leq
|\mathcal{X}|$, $|\mathcal{V}|\leq |\mathcal{X}|$,
$|\mathcal{W}|\leq |\mathcal{X}|+4$ and $H(X|UVW)=0$ is identical to
Marton's inner bound.

One of the main contributions of this paper is the perturbation technique. At the heart of this technique lies the following observation:
consider an arbitrary set of finite random variables $X_1, X_2, ...,
X_n$ jointly distributed according to $p_0(x_1, x_2, ..., x_n)$. One
can represent a perturbation of this joint distribution by a vector
consisting of the first derivative of the individual probabilities
$p_0(x_1, x_2, ..., x_n)$ for all values of $x_1$, $x_2$, ...,
$x_n$. We however suggest the following perturbation that can be
represented by a real valued random variable, $L$, jointly
distributed by $X_1, X_2, ..., X_n$ and satisfying
$\mathbb{E}[L]=0$, $\big|\mathbb{E}[L|X_1=x_1, X_2=x_2,...,
X_n=x_n]\big|<\infty$ for all values of $x_1$, $x_2$, ..., $x_n$:
$$p_{\epsilon}(\widehat{X}_1=x_1,...,\widehat{X}_n=x_n)= p_0(X_1=x_1,...,X_n=x_n)\cdot \big(1+\epsilon\cdot
\mathbb{E}[L|X_1=x_1,...,X_n=x_n]\big),$$ where $\epsilon$ is a real
number in some interval $[-\overline{\epsilon}_1,
\overline{\epsilon}_2]$. Random variable $L$ is a canonical way of
representing the direction of perturbation since given any subset of
indices $I\subset \{1,2,3,...,n\}$, one can verify that the
following equation for the marginal distribution of random variables
$\widehat{X}_i$ for $i\in I$:
$$p_{\epsilon}(\widehat{X}_{i\in I}=x_{i\in I})= p_0(X_{i\in I}=x_{i\in I})\cdot \big(1+\epsilon\cdot
\mathbb{E}[L|X_{i\in I}=x_{i\in I}]\big).$$ Furthermore for any set
of indices $I\subset \{1,2,3,...,n\}$, the second derivative of the
joint entropy of random variables $\widehat{X}_i$ for $i\in I$ as a
function of $\epsilon$ is related to the problem of MMSE estimation
of $L$ from $X_{i\in I}$:
$$ \frac{\partial^2}{\partial \epsilon^2}H(\widehat{X}_{i\in
I})\mid_{\epsilon=0}=-\log e\cdot
\mathbb{E}\big[\mathbb{E}[L|X_{i\in I}]^2\big].$$ Lemma \ref{Lemma3}
describes a generic version of the above identity that relates the
second derivative of the Shannon entropy of a discrete random
variable to the corresponding Fisher information. This identity is
to best of our knowledge new. It is repeatedly invoked in our proofs
to compute the second derivative of various expressions.

It is known that Marton's inner bound coincides with the outer bound of Nair and El Gamal for the degraded, less noisy, more capable, and
semi-deterministic broadcast channels. Nair and Zizhou showed that
Marton's inner bound and the recent outer bound of Nair and El Gamal
are different for a BSSC channel with parameter $\frac{1}{2}$ if a
certain conjecture holds\footnote{The conjecture is as follows:
\cite[Conjecture 1]{NairZizhou}: Given any five random variables $U,
V, X, Y, Z$ satisfying $I(UV;YZ|X)=0$, the inequality
$I(U;Y)+I(V;Z)-I(U;V)\leq \max\big(I(X;Y), I(X;Z)\big)$ holds
whenever $X$, $Y$ and $Z$ are binary random variables and the
channel $p(y,z|x)$ is BSSC with parameter $\frac{1}{2}$.}. In this
paper, we provide examples of broadcast channels for which the two
bounds do not match. Since the original submission of this paper, Nair, Wang and Geng \cite{NairWangGeng} showed that the inequality $I(U;Y)+I(V;Z)-I(U;V)\leq \max\big(I(X;Y), I(X;Z)\big)$ holds for all binary input broadcast channels. The authors employ a generalized version of the perturbation method introduced in this paper that also allows for additive perturbations. The authors of \cite{GohariElGamalAnantharam} prove various results that help to restrict the search space for computing the sum-rate for Marton's inner bound. 

The outline of this paper is as follows. In section
\ref{Section:Definition}, we introduce the basic notation and
definitions we use. Section \ref{Section:MainResults} contains the main results of the
paper followed by section \ref{Section:Proofs} which gives formal
proofs for the results. Section \ref{Section:Description} describes the new ideas, and appendices complete the proof of theorems
from section \ref{Section:Proofs}.
\section{Definitions and Notation}\label{Section:Definition}
Let $\mathbb{R}$ denote the set of real numbers. All the logarithms
throughout this paper are in base two, unless stated otherwise. Let
$\mathcal{C}(q(y, z|x))$ denote the capacity region of the broadcast
channel $q(y,z|x)$. We use $X_{1:k}$ to denote $(X_{1}, X_{2},...,
X_{k})$; similarly we use $Y_{1:k}$ and $Z_{1:k}$ to denote $(Y_{1},
Y_{2},..., Y_{k})$ and $(Z_{1}, Z_{2},..., Z_{k})$ respectively.

\begin{definition}\label{definition:4N} \label{def:downset}
For two vectors $\overrightarrow{v_1}$ and $\overrightarrow{v_2}$ in
$\mathbb{R}^d$, we say $\overrightarrow{v_1}\geq
\overrightarrow{v_2}$ if and only if each coordinate of
$\overrightarrow{v_1}$ is greater than or equal to the corresponding
coordinate of $\overrightarrow{v_2}$. For a set $A \subset
\mathbb{R}^d$, the down-set $\Delta(A)$ is defined as:
$\Delta(A)=\{\overrightarrow{v} \in \mathbb{R}^d:\
\overrightarrow{v} \leq \overrightarrow{w} \mbox{ for some }
\overrightarrow{w} \in A\}$.\end{definition}

\begin{definition}\label{Definition1} Let
$\mathcal{C}_{M}(q(y, z|x))$ denote Marton's inner bound on the
channel $q(y,z|x)$. $\mathcal{C}_{M}(q(y, z|x))$ is defined as the
union over non-negative triples $(R_0, R_1, R_2)$ satisfying
equations \ref{eqn:DefOfMKG1}, \ref{eqn:DefOfMKG2},
\ref{eqn:DefOfMKG3}, and \ref{eqn:DefOfMKG4} over random variables
$U,V,W,X,Y,Z$, having the joint distribution
$p(u,v,w,x,y,z)=p(u,v,w,x)q(y,z|x)$. Please note that the auxiliary
random variables $U, V$ and $W$ may be discrete or continuous random
variables.
\end{definition}

\begin{definition} The region ${\mathcal{C}}_{M}^{S_u, S_v, S_w}(q(y,
z|x))$ is defined as the union of non-negative triples $(R_0, R_1,
R_2)$ satisfying equations \ref{eqn:DefOfMKG1}, \ref{eqn:DefOfMKG2},
\ref{eqn:DefOfMKG3} and \ref{eqn:DefOfMKG4}, over discrete random
variables $U,V,W,X,Y,Z$ satisfying the cardinality bounds
$|\mathcal{U}|\leq S_u$, $|\mathcal{V}|\leq S_v$ and
$|\mathcal{W}|\leq S_w$, and having the joint distribution
$p(u,v,w,x,y,z)=p(u,v,w,x)q(y,z|x)$. Note that
${\mathcal{C}}_{M}^{S_u, S_v, S_w}(q(y, z|x))\subset
{\mathcal{C}}_{M}^{S'_u, S'_v, S'_w}(q(y, z|x))$ whenever $S_u\leq
S'_u$, $S_v\leq S'_v$ and $S_w\leq S'_w$.
\end{definition}

\begin{definition} Let $\mathscr{L}(q(y, z|x))$ be equal to ${\mathcal{C}}_{M}^{|\mathcal{X}|,
|\mathcal{X}|, |\mathcal{X}|+4}(q(y, z|x))$.
\end{definition}

\begin{definition}The region ${\mathscr{C}}(q(y, z|x))$ is defined as the union over
discrete random variables $U,V,W,X,Y,Z$ satisfying the cardinality
bounds $|\mathcal{U}|\leq |\mathcal{X}|$, $|\mathcal{V}|\leq
|\mathcal{X}|$ and $|\mathcal{W}|\leq |\mathcal{X}|+4$, and having
the joint distribution $p(u,v,w,x,y,z)=p(u,v,w,x)q(y,z|x)$ for which
$H(X|UVW)=0$, of non-negative triples $(R_0, R_1, R_2)$ satisfying
equations \ref{eqn:DefOfMKG1}, \ref{eqn:DefOfMKG2},
\ref{eqn:DefOfMKG3} and \ref{eqn:DefOfMKG4}. Please note that the
definition of ${\mathscr{C}}(q(y, z|x))$ differs from that of
$\mathscr{L}(q(y, z|x))$ since we have imposed the extra constraint
$H(X|UVW)=0$ on the auxiliaries. $\mathscr{C}(q(y, z|x))$ is a
\emph{computable} subset of the region ${\mathcal{C}}_{M}(q(y,
z|x))$. 
\end{definition}

\begin{definition} Given broadcast channel $q(y,z|x)$, let
${\mathcal{C}}_{NE}(q(y, z|x))$ denote the union over random
variables $U,V,W,X,Y,Z$, having the joint distribution
$p(u,v,w,x,y,z)=p(u)p(v)p(w|u,v)p(x|u,v,w)q(y,z|x)$, of
  \begin{eqnarray*}
    R_0, R_1, R_2&\geq& 0;\\
    R_0&\leq& \min(I(W;Y),I(W;Z));\\
    R_0+R_1&\leq& I(UW;Y);\\
    R_0+R_2&\leq& I(VW;Z);\\
    R_0+R_1+R_2&\leq&
I(UW;Y)+I(V;Z|UW);\\
    R_0+R_1+R_2&\leq&
I(VW;Z)+I(U;Y|VW).
  \end{eqnarray*}
${\mathcal{C}}_{NE}(q(y, z|x))$ is shown in \cite{Nair-El-Gamal} to
be an outer bound to the capacity region of the broadcast channel. This outer bound matches the best known outer bound discussed in \cite{NairNoteOuterBound} when $R_0=0$. An  alternative characterization of the set of triples $(0, R_1, R_2)$ in ${\mathcal{C}}_{NE}(q(y, z|x))$ is as follows \cite{NairZizhou}:
the union over random variables $U,V,X,Y,Z$ having the joint
distribution $p(u,v,x,y,z)=p(u,v,x)q(y,z|x)$, of
  \begin{eqnarray*}
    R_1, R_2&\geq& 0;\\
    R_1&\leq& I(U;Y);\\
    R_2&\leq& I(V;Z);\\
    R_1+R_2&\leq&
I(U;Y)+I(V;Z|U);\\
    R_1+R_2&\leq&
I(V;Z)+I(U;Y|V).
  \end{eqnarray*}
\end{definition}
%

\section{Statement of results}\label{Section:MainResults}
\begin{theorem}\label{Thm:Thm1} For any arbitrary broadcast channel $q(y,
z|x)$, the closure of ${\mathcal{C}}_{M}(q(y, z|x))$ is equal to
$\mathscr{C}(q(y, z|x))$.
\end{theorem}
\begin{corollary} ${\mathcal{C}}_{M}(q(y, z|x))$ is closed since $\mathscr{C}(q(y, z|x))$ is also a subset of ${\mathcal{C}}_{M}(q(y, z|x))$.\end{corollary}
\begin{theorem}\label{Thm:Thm2} There are broadcast channels for which Marton's inner bound and the recent outer bound of Nair and El Gamal do not match.
\end{theorem}

\section{Description of the main technique}\label{Section:Description}
In this section, we demonstrate the main idea of the paper. In order
to show the essence of the proof while avoiding the unnecessary
details, we consider a simpler problem that is different from the
problem at hand, although it will be used in the later proofs. 

Given a broadcast channel $q(y,z|x)$ and an input distribution
$p(x)$, let us consider the problem of finding the supremum of
\begin{eqnarray*}&I(U;Y)+I(V;Z)-I(U;V)+\lambda I(U;Y)+\gamma
I(V;Z)\end{eqnarray*} over all joint distributions
$p(uv|x)p(x)q(y,z|x)$ where $\lambda$ and $\gamma$ are arbitrary
non-negative reals, and auxiliary random variables $U$, $V$ have
alphabet sets satisfying $|\mathcal{U}|\leq S_u$ and
$|\mathcal{V}|\leq S_v$ for some natural numbers $S_u$ and $S_v$.
For this problem, we would like to show that it suffices to take the
maximum over random variables $U$ and $V$ with the cardinality
bounds of $\min(|\mathcal{X}|, S_u)$ and $\min(|\mathcal{X}|, S_v)$.
It suffices to prove the following lemma:

\begin{lemma}\label{Lemma-1} Given an arbitrary broadcast channel $q(y,z|x)$, an
arbitrary input distribution $p(x)$, non-negative reals $\lambda$
and $\gamma$, and natural numbers $S_u$ and $S_v$ where
$S_u>|\mathcal{X}|$ the following holds:
\begin{eqnarray*}&\sup_{UV\rightarrow X\rightarrow YZ; |\mathcal{U}|\leq S_u; |\mathcal{V}|\leq S_v}I(U;Y)+I(V;Z)-I(U;V)+\lambda I(U;Y)+\gamma
I(V;Z)=\\&I(\widehat{U};\widehat{Y})+I(\widehat{V};\widehat{Z})-I(\widehat{U};\widehat{V})+\lambda
I(\widehat{U};\widehat{Y})+\gamma
I(\widehat{V};\widehat{Z}),\end{eqnarray*} where random variables
$\widehat{U}, \widehat{V}, \widehat{X}, \widehat{Y}, \widehat{Z}$
satisfy the following properties: the Markov chain
$\widehat{U}\widehat{V}\rightarrow \widehat{X}\rightarrow
\widehat{Y}\widehat{Z}$ holds; the joint distribution of
$\widehat{X}, \widehat{Y}, \widehat{Z}$ is the same as the joint
distribution of $X,Y,Z$, and furthermore
$|\mathcal{\widehat{U}}|<S_u$, $|\mathcal{\widehat{V}}|\leq S_v$.
\end{lemma}

\subsection{Proof based on the perturbation method}
Since the cardinalities of $U$ and $V$ are bounded,
one can show that the supremum of $I(U;Y)+I(V;Z)-I(U;V)+\lambda
I(U;Y)+\gamma I(V;Z)$ is a maximum\footnote{Since the ranges of all
the random variables are finite and the conditional
mutual information function is continuous, the set of admissible
joint probability distributions $p(u,v,x,y,z)$ where $I(UV;YZ|X)=0$
and $p(y,z,x)=q(y,z|x)p(x)$ will be a compact set (when viewed as a
subset of the Euclidean space). The fact that mutual information
function is continuous implies that the union over random variables
$U,V,X,Y,Z$ satisfying the cardinality bounds, having the joint
distribution $p(u,v,x,y,z)=p(u,v|x)p(x)q(y,z|x)$, of
$I(U;Y)+I(V;Z)-I(U;V)+\lambda I(U;Y)+\gamma I(V;Z)$ is a compact
set, and thus closed.}, and is obtained at some joint distribution
$p_0(u,v,x,y,z)=p_0(u,v,x)q(y,z|x)$. If $|\mathcal{U}|<S_u$, one can
finish the proof by setting $(\widehat{U}, \widehat{V}, \widehat{X},
\widehat{Y}, \widehat{Z}) =(U, V, X, Y, Z)$. One can also easily
show the existence of appropriate $(\widehat{U}, \widehat{V},
\widehat{X}, \widehat{Y}, \widehat{Z})$ if $p(u)=0$ for some $u\in
\mathcal{U}$. Therefore assume that $|\mathcal{U}|=S_u$ and
$p(u)\neq 0$ for all $u\in \mathcal{U}$. Take an arbitrary non-zero
function $L:\mathcal{U}\times \mathcal{V}\times
\mathcal{X}\rightarrow \mathbb{R}$ where $\mathbb{E}[L(U,V,X)|X]$=0.
Let us then perturb the joint distribution of $U,V,X,Y,Z$ by
defining random variables
$\widehat{U},\widehat{V},\widehat{X},\widehat{Y}$ and $\widehat{Z}$
distributed according to 
\begin{eqnarray*}&p_{\epsilon}(\widehat{U}=u, \widehat{V}=v,
\widehat{X}=x, \widehat{Y}=y,\widehat{Z}=z)=\\&p_{0}(U=u, V=v, X=x,
Y=y,Z=z)\cdot\big(1+\epsilon\cdot
\mathbb{E}[L(U,V,X)|U=u,V=v,X=x,Y=y,Z=z]\big),\end{eqnarray*}\normalsize
or equivalently according to
\begin{eqnarray*}&p_{\epsilon}(\widehat{U}=u, \widehat{V}=v,
\widehat{X}=x, \widehat{Y}=y,\widehat{Z}=z)=\\&p_{0}(U=u, V=v, X=x,
Y=y,Z=z)\big(1+\epsilon\cdot L(u,v,x)\big)=\\&p_{0}(U=u, V=v, X=x)q(
Y=y,Z=z|X=x)\big(1+\epsilon\cdot L(u,v,x)\big).\end{eqnarray*} The
parameter $\epsilon$ is a real number that can take values in
$[-\overline{\epsilon}_1, \overline{\epsilon}_2]$ where
$\overline{\epsilon}_1$ and $\overline{\epsilon}_2$ are some
positive reals representing the maximum and minimum values of
$\epsilon$, i.e. $\min_{u,v,x}1-\overline{\epsilon}_1\cdot
L(u,v,x)=\min_{u,v,x}1+\overline{\epsilon}_2\cdot L(u,v,x)=0$. Since
$L$ is a function of $U$, $V$ and $X$ only, for any value of
$\epsilon$, the Markov chain
$\widehat{U}\widehat{V}\rightarrow\widehat{X}\rightarrow\widehat{Y}\widehat{Z}$
holds, and $p(\widehat{Y}=y,\widehat{Z}=z|\widehat{X}=x)$ is equal
to $q(Y=y,Z=z|X=x)$ for all $x,y,z$ where $p(X=x)>0$. Furthermore
 $\mathbb{E}[L(U,V,X)|X]=0$ implies that the marginal
distribution of $X$ is preserved by this perturbation. This is
because
\begin{eqnarray*}&p_{\epsilon}(\widehat{X}=x)=p_{0}(X=x)\cdot\big(1+\epsilon\cdot
\mathbb{E}[L(U,V,X)|X=x]\big).\end{eqnarray*} This further implies
that the marginal distributions of $Y$ and $Z$ are also fixed.
\footnote{The terms $\mathbb{E}[L(U,V,X)|Y]=0$ and
$\mathbb{E}[L(U,V,X)|Z]=0$ must be zero if
$\mathbb{E}[L(U,V,X)|X]=0$}

The expression
$I(\widehat{U};\widehat{Y})+I(\widehat{V};\widehat{Z})-I(\widehat{U};\widehat{V})+\lambda
I(\widehat{U};\widehat{Y})+\gamma I(\widehat{V};\widehat{Z})$ as a
function of $\epsilon$ achieves its maximum at $\epsilon=0$ (by our
assumption). Therefore its first derivative at $\epsilon=0$ should
be zero, and its second derivative should be less than or equal to
zero. We use the following lemma to compute the first derivative and the second derivative of the above expression.
\begin{lemma}\label{Lemma3} Given any finite random variable $X$, and real
valued random variable $L$ where
$\big|\mathbb{E}[L|X=x]\big|<\infty$ for all $x\in \mathcal{X}$, and
$\mathbb{E}[L]=0$, let random variable $\widehat{X}$ be defined on
the same alphabet set as $X$ according to
$p_{\epsilon}(\widehat{X}=x)= p_0(X=x)\cdot \big(1+\epsilon\cdot
\mathbb{E}[L|X=x]\big),$ where $\epsilon$ is a real number in the
interval $[-\overline{\epsilon}_1, \overline{\epsilon}_2]$.
$\overline{\epsilon}_1$ and $\overline{\epsilon}_2$ are positive
reals for which $\min_{x}1-\overline{\epsilon}_1\cdot
\mathbb{E}[L|X=x]\geq 0$ and $\min_x1+\overline{\epsilon}_2\cdot
\mathbb{E}[L|X=x]\geq 0$ hold. Then
\begin{enumerate}
  \item $H(\widehat{X})\mid_{\epsilon=0}=H(X)$, and $\frac{\partial}{\partial
  \epsilon}H(\widehat{X})\mid_{\epsilon=0}=H_L(X)$ where $H_L(X)$ is defined as $H_L(X)=\sum_{x\in\mathcal{X}}p(X=x)\mathbb{E}[L|X=x]\log\frac{1}{p(X=x)}$ 
 for any finite random variable $X$ and real valued random variable $L$ where $\big|\mathbb{E}[L|X=x]\big|<\infty$ for all $x\in \mathcal{X}$.
  \item $\forall \epsilon \in (-\overline{\epsilon}_1, \overline{\epsilon}_2)$, $\frac{\partial^2}{\partial \epsilon^2}H(\widehat{X})=-\log e\cdot
\mathbb{E}\big[\frac{\mathbb{E}[L|X]^2}{1+\epsilon\cdot
\mathbb{E}[L|X]}\big]=-\log (e)\cdot I(\epsilon)$ where the Fisher
Information $I(\epsilon)$ is defined as
$I(\epsilon)=\sum_{x}\bigg(\frac{\partial}{\partial
\epsilon}\log_e\big(p_{\epsilon}(\widehat{X}=x)\big)\bigg)^2p_{\epsilon}(\widehat{X}=x).$
  In particular $\frac{\partial^2}{\partial
  \epsilon^2}H(\widehat{X})\mid_{\epsilon=0}=-\log e\cdot
\mathbb{E}\big[\mathbb{E}[L|X]^2\big]$.
\item $H(\widehat{X})=H(X)+\epsilon H_L(X)-\mathbb{E}\big[r\big(\epsilon\cdot
\mathbb{E}[L|X]\big)\big]$ where $r(x)=(1+x)\log(1+x)$.
\end{enumerate}
\end{lemma}
Using the above lemma, one can compute the first derivative
and set it to zero, and thereby get the following equation:
\begin{eqnarray}&I_L(U;Y)+I_L(V;Z)-I_L(U;V)+\lambda I_L(U;Y)+\gamma
I_L(V;Z)=0,\label{eqn:IDE1}\end{eqnarray}
where $I_L(X;Y)=H_L(X)-H_L(X|Y)=H_L(Y)-H_L(Y|X)$, $H_L(X|Y)=\sum_{y\in \mathcal{Y}}p(Y=y)H_L(X|Y=y),$ and $H_L(X|Y=y)=\sum_{x\in
\mathcal{X}}p(X=x|Y=y)\mathbb{E}[L|X=x,Y=y]\log\frac{1}{p(X=x|Y=y)}$ for any finite random variables $X$ and $Y$ and real valued random
variable $L$ where $\big|\mathbb{E}[L|X=x, Y=y]\big|<\infty$ for all $x\in \mathcal{X}$ and $y \in \mathcal{Y}$.

In order to compute the
second derivative, one can expand the expression through entropy terms
and use Lemma \ref{Lemma3} to compute the second derivative for each
term. We can use the assumption that $\mathbb{E}[L(U,V,X)|X]=0$
(which implies $\mathbb{E}[L(U,V,X)|Y]=0$ and
$\mathbb{E}[L(U,V,X)|Z]=0$) to simplify the expression. In
particular the second derivative of $H(\widehat{Y})$ and
$H(\widehat{Z})$ at $\epsilon=0$ would be equal to zero (as the
marginal distributions of $Y$ and $Z$ are preserved under the
perturbation), the second derivative of $I(\widehat{U};\widehat{Y})$
at $\epsilon=0$ will be equal to $-\log e\cdot
\mathbb{E}[\mathbb{E}[L(U,V,X)|U]^2]+\log e\cdot
\mathbb{E}[\mathbb{E}[L(U,V,X)|UY]^2]$, the second derivative of
$I(\widehat{V};\widehat{Z})$ at $\epsilon=0$ will be equal to $-\log
e\cdot \mathbb{E}[\mathbb{E}[L(U,V,X)|V]^2]+\log e\cdot
\mathbb{E}[\mathbb{E}[L(U,V,X)|VZ]^2]$, and the second derivative of
$-I(\widehat{U};\widehat{V})$ at $\epsilon=0$ will be equal to
$+\log e\cdot \mathbb{E}[\mathbb{E}[L(U,V,X)|U]^2]+\log e\cdot
\mathbb{E}[\mathbb{E}[L(U,V,X)|V]^2]-\log e\cdot
\mathbb{E}[\mathbb{E}[L(U,V,X)|UV]^2]$. Note that the second
derivatives of $I(\widehat{U};\widehat{Y})$ and
$I(\widehat{V};\widehat{Z})$ are always non-negative. Since the
second derivative of the expression
$I(\widehat{U};\widehat{Y})+I(\widehat{V};\widehat{Z})-I(\widehat{U};\widehat{V})+\lambda
I(\widehat{U};\widehat{Y})+\gamma I(\widehat{V};\widehat{Z})$ at
$\epsilon=0$ must be non-positive, the second derivative of
$I(\widehat{U};\widehat{Y})+I(\widehat{V};\widehat{Z})-I(\widehat{U};\widehat{V})$
must be non-positive at $\epsilon=0$. The second derivative of the
latter expression is equal to $+\log e\cdot
\mathbb{E}[\mathbb{E}[L(U,V,X)|UY]^2]+\log e\cdot
\mathbb{E}[\mathbb{E}[L(U,V,X)|VZ]^2]-\log e\cdot
\mathbb{E}[\mathbb{E}[L(U,V,X)|UV]^2]$. Hence we conclude that for
any non-zero function $L:\mathcal{U}\times \mathcal{V}\times
\mathcal{X}\rightarrow \mathbb{R}$ where $\mathbb{E}[L(U,V,X)|X]=0$
we must have:
\begin{eqnarray}\label{eqnation:lastchanges}&\mathbb{E}[\mathbb{E}[L(U,V,X)|UY]^2]+
\mathbb{E}[\mathbb{E}[L(U,V,X)|VZ]^2]-
\mathbb{E}[\mathbb{E}[L(U,V,X)|UV]^2]\leq 0.\end{eqnarray}

Next, take an arbitrary non-zero function $L':\mathcal{U}\rightarrow
\mathbb{R}$ where $\mathbb{E}[L'(U)|X]=0$. Since
$|\mathcal{U}|=S_u>|\mathcal{X}|$, such a non-zero function $L'$
exists. Note that the direction of perturbation $L'$ being only a
function of $U$ implies that
\begin{eqnarray*}&p_{\epsilon}(\widehat{U}=u, \widehat{V}=v,
\widehat{X}=x,
\widehat{Y}=y,\widehat{Z}=z)=\\&p_{\epsilon}(\widehat{U}=u)p_0(V=v,
X=x, Y=y,Z=z|U=u)\end{eqnarray*} In other words, the perturbation
only changes the marginal distribution of $U$, but preserves the
conditional distribution of $p_0(V=v, X=x, Y=y,Z=z|U=u)$.

Note that
\begin{eqnarray*}&\mathbb{E}[\mathbb{E}[L'(U)|UV]^2]=\mathbb{E}[\mathbb{E}[L'(U)|UY]^2]=\mathbb{E}[L'(U)^2].\end{eqnarray*}
This implies that $\mathbb{E}[\mathbb{E}[L'(U)|VZ]^2]$ should be
non-positive. But this can happen only when
$\mathbb{E}[L'(U)|VZ]=0$. Therefore any arbitrary function
$L':\mathcal{U}\rightarrow \mathbb{R}$ where $\mathbb{E}[L'(U)|X]=0$
must also satisfy $\mathbb{E}[L'(U)|VZ]=0$. In other words, any
arbitrary direction of perturbation $L'$ that is a function of $U$
and preserves the marginal distribution of $X$, must also preserve
the marginal distribution of $VZ$.\footnote{Note that
$p_{\epsilon}(\widehat{V}=v, \widehat{Z}=z )=p_{0}(V=v,
Z=z)\cdot\big(1+\epsilon\cdot \mathbb{E}[L(U,V,X)|V=v,
Z=z]\big)=p_{0}(V=v, Z=z)$.}

We next show that the expression
$I(\widehat{U};\widehat{Y})+I(\widehat{V};\widehat{Z})-I(\widehat{U};\widehat{V})+\lambda
I(\widehat{U};\widehat{Y})+\gamma I(\widehat{V};\widehat{Z})$ as a
function of $\epsilon$ is constant.\footnote{The authors would like
to thank Chandra Nair for suggesting this shortcut to simplify the
original proof.} Using the last part of Lemma \ref{Lemma3}, one can
write:
\begin{eqnarray}&I(\widehat{U};\widehat{Y})=I(U;Y)+\epsilon\cdot
I_L(\widehat{U};\widehat{Y})-\mathbb{E}\big[r\big(\epsilon\cdot
\mathbb{E}[L|U]\big)\big]-\mathbb{E}\big[r\big(\epsilon\cdot
\mathbb{E}[L|Y]\big)\big]+\mathbb{E}\big[r\big(\epsilon\cdot
\mathbb{E}[L|UY]\big)\big]=\nonumber\\&I(U;Y)+\epsilon\cdot
I_L(\widehat{U};\widehat{Y}),\label{eqn:IDE2}\end{eqnarray} where
$r(x)=(1+x)\log(1+x)$. Equation (\ref{eqn:IDE2}) holds because
$\mathbb{E}[L|Y]=0$ and $\mathbb{E}[L|U]=\mathbb{E}[L|UY]$.
Similarly using the last part of Lemma \ref{Lemma3}, one can write:
\begin{eqnarray}&I(\widehat{U};\widehat{V})=I(U;V)+\epsilon\cdot
I_L(\widehat{U};\widehat{V})-\mathbb{E}\big[r\big(\epsilon\cdot
\mathbb{E}[L|U]\big)\big]-\mathbb{E}\big[r\big(\epsilon\cdot
\mathbb{E}[L|V]\big)\big]+\mathbb{E}\big[r\big(\epsilon\cdot
\mathbb{E}[L|UV]\big)\big]=\nonumber\\&I(U;V)+\epsilon\cdot
I_L(\widehat{U};\widehat{V})\label{eqn:IDE2N}\end{eqnarray} where
$r(x)=(1+x)\log(1+x)$. Equation (\ref{eqn:IDE2N}) holds because
$\mathbb{E}[L|V]=0$ and $\mathbb{E}[L|U]=\mathbb{E}[L|UV]$. One can
similarly show that the term $I(\widehat{V};\widehat{Z})$ can be
written as $I(V;Z)+\epsilon\cdot I_L(\widehat{V};\widehat{Z})=0$.
Therefore the expression
$I(\widehat{U};\widehat{Y})+I(\widehat{V};\widehat{Z})-I(\widehat{U};\widehat{V})+\lambda
I(\widehat{U};\widehat{Y})+\gamma I(\widehat{V};\widehat{Z})$ as a
function of $\epsilon$ is equal to
\begin{eqnarray}&I(U;Y)+I(V;Z)-I(U;V)+\lambda I(U;Y)+\gamma
I(V;Z)+\nonumber\\&\epsilon\cdot
\big(I_L(U;Y)+I_L(V;Z)-I_L(U;V)+\lambda I_L(U;Y)+\gamma
I_L(V;Z)\big).\label{eqn:IDE3}\end{eqnarray} Equation
(\ref{eqn:IDE1}) implies that this expression is equal to
$I(U;Y)+I(V;Z)-I(U;V)+\lambda I(U;Y)+\gamma I(V;Z)$.

Therefore the expression
$I(\widehat{U};\widehat{Y})+I(\widehat{V};\widehat{Z})-I(\widehat{U};\widehat{V})+\lambda
I(\widehat{U};\widehat{Y})+\gamma I(\widehat{V};\widehat{Z})$ as a
function of $\epsilon$ is constant. Since the function $L'$ is
non-zero, setting $\epsilon=-\overline{\epsilon}_1$ or
$\epsilon=\overline{\epsilon}_2$ will result in a marginal
distribution on $\widehat{U}$ with a smaller support than $U$ since
the marginal distribution of $U$ is being perturbed as follows:
\begin{eqnarray*}&p_{\epsilon}(\widehat{U}=u)=p_{0}(U=u)\cdot\big(1+\epsilon
L'(u)\big).\end{eqnarray*} This perturbation does not increase the
support and would decrease it by at least one when $\epsilon$ is at
its maximum or minimum, i.e. when $\epsilon=-\overline{\epsilon}_1$
or $\epsilon=\overline{\epsilon}_2$. Therefore one is able to define
a random variable with a smaller cardinality as that of $U$ while
leaving the value of $I(U;Y)+I(V;Z)-I(U;V)+\lambda I(U;Y)+\gamma
I(V;Z)$ unaffected.

\emph{Discussion:} Aside from establishing cardinality bounds, the
above argument implies that if the maximum of
$I(U;Y)+I(V;Z)-I(U;V)+\lambda I(U;Y)+\gamma I(V;Z)$ is obtained at
some joint distribution $p_0(u,v,x,y,z)=p_0(u,v,x)q(y,z|x)$,
equations \ref{eqn:IDE1} and \ref{eqnation:lastchanges} must hold
for any non-zero function $L:\mathcal{U}\times \mathcal{V}\times
\mathcal{X}\rightarrow \mathbb{R}$ where $\mathbb{E}[L(U,V,X)|X]=0$.
The proof used these properties to a limited extent.

\subsection{Alternative proof}
In this subsection we provide an alternative proof for Lemma \ref{Lemma-1}. 
Assume that the maximum of $I(U;Y)+I(V;Z)-I(U;V)+\lambda
I(U;Y)+\gamma I(V;Z)$ is obtained at some joint distribution
$p_0(u,v,x,y,z)=p_0(u,v,x)q(y,z|x)$. Without loss of generality we can assume that $p(u)>0$ for all $u\in \mathcal{U}$. Let us fix $p_0(v,x|u)q(y,z|x)$ and vary the marginal distribution of $U$ in such a way that the marginal distribution of $X$ is preserved. In other words, we consider the set of p.m.f's $q(u)$ satisfying $\sum_{u,v}q(u)p_0(v,x|u)=p_0(x)$ for all $x\in \mathcal{X}$. We can then view the expression $I(U;Y)+I(V;Z)-I(U;V)+\lambda
I(U;Y)+\gamma I(V;Z)$ as a function of a p.m.f $q(u)$ defined on $\mathcal{U}$. Here $U,V,X,Y,Z$ are jointly distributed as $q(u)p_0(v,x|u)q(y,z|x)$. We claim that $I(U;Y)+I(V;Z)-I(U;V)+\lambda
I(U;Y)+\gamma I(V;Z)$ is convex function over $q(u)$. To see this note that $I(U;Y)+I(V;Z)-I(U;V)=H(Y)-H(Y|U)-H(V|Z)+H(V|U)$. Since the marginal distribution of $X$ is preserved, $H(Y)$ is fixed. The term $-H(Y|U)+H(V|U)$ is linear in $q(u)$, and $-H(V|Z)$ is convex in $q(u)$. Therefore $I(U;Y)+I(V;Z)-I(U;V)$ is a convex function over $q(u)$. Next, note that $\lambda
I(U;Y)=\lambda H(Y)-\lambda H(Y|U)$ is linear in $q(u)$, and $\gamma I(V;Z)=\gamma H(Z)-\gamma H(Z|V)$ is convex in $q(u)$. The latter is because the marginal distribution of $X$ is preserved and hence $H(Z)$ is fixed. All in all, we can conclude that  $I(U;Y)+I(V;Z)-I(U;V)+\lambda
I(U;Y)+\gamma I(V;Z)$ is convex in $q(u)$. This implies that it will have a maximum at the extreme points of the domain. We claim that any extreme point of the domain corresponds to a p.m.f $q(u)$ with support at most $|\mathcal{X}|$. This completes the proof. The domain of the function is the polytope formed by the set of vectors $(q(u): u\in \mathcal{U})$ satisfying the following constraints
\begin{align*}
&q(u)\geq 0,~~~~ \forall u\in \mathcal{U}\\&
\sum_{u\in \mathcal{U}}q(u)=1\\&
\sum_{u,v}q(u)p_0(v,x|u)=p_0(x),~~~~ \forall x\in \mathcal{X}
\end{align*}
Note that the equation $\sum_{u\in \mathcal{U}}q(u)=1$ is redundant and implied by the others because 
$1=\sum_{x}p_0(x)=\sum_{x}\sum_{u,v}q(u)p_0(v,x|u)=\sum_{u}\sum_{v,x}q(u)p_0(v,x|u)=\sum_{u}q(u)$. Thus, we can describe the domain of the function by
\begin{align*}
&q(u)\geq 0,~~~~ \forall u\in \mathcal{U}\\&
\sum_{u,v}q(u)p_0(v,x|u)=p_0(x),~~~~ \forall x\in \mathcal{X}
\end{align*}
Any extreme point of this polytope must lie on at least $|\mathcal{U}|$ hyperplanes because the polytope lies in $\mathbb{R}^{|\mathcal{U}|}$. Because there are $|\mathcal{X}|$ equations of the type $\sum_{u,v}q(u)p_0(v,x|u)=p_0(x)$, any extreme point has to pick up at least $|\mathcal{U}|-|\mathcal{X}|$ equation of the type $q(u)\geq 0$. This implies that $q(u)=0$ for at least $|\mathcal{U}|-|\mathcal{X}|$ different values of $u\in \mathcal{U}$. Therefore the support of any extreme point must be less than or equal to $|\mathcal{U}|-(|\mathcal{U}|-|\mathcal{X}|)=|\mathcal{X}|$.

\section{Proofs}\label{Section:Proofs}
\begin{proof}[Proof of Theorem \ref{Thm:Thm1}]
We begin by showing that for any natural numbers $S_u, S_v, S_w$, one has ${\mathcal{C}}_{M}^{S_u, S_v, S_w}(q(y, z|x))\subset
{\mathcal{C}}_{M}^{|\mathcal{X}|, |\mathcal{X}|, |\mathcal{X}|+4}(q(y, z|x))={\mathscr{L}}(q(y, z|x))$. This is proved in two steps:
\begin{enumerate}
  \item ${\mathcal{C}}_{M}^{S_u, S_v, S_w}(q(y, z|x))\subset
{\mathcal{C}}_{M}^{S_u, S_v, |\mathcal{X}|+4}(q(y, z|x))$.
  \item ${\mathcal{C}}_{M}^{S_u, S_v, |\mathcal{X}|+4}(q(y, z|x))\subset
{\mathcal{C}}_{M}^{|\mathcal{X}|, |\mathcal{X}|, |\mathcal{X}|+4}(q(y, z|x))$.
\end{enumerate}
The first step that imposes a cardinality bound on the alphabet set of $W$ follows just from a standard application of the strengthened Carath\'{e}odory theorem of Fenchel and is left to the reader. The difficult part is the second step. To show this it suffices to prove more generally that \begin{eqnarray}&{\mathcal{C}}_{M-I}^{S_u, S_v, |\mathcal{X}|+4}(q(y, z|x))\subset
{\mathcal{C}}_{M-I}^{|\mathcal{X}|, |\mathcal{X}|, |\mathcal{X}|+4}(q(y, z|x))\label{eqn:MP1}\end{eqnarray} where ${\mathcal{C}}_{M-I}^{S_u, S_v, S_w}$ is defined as the union of real four tuples $(R'_1, R'_2, R'_3, R'_4)$ satisfying
  \begin{eqnarray}
    R'_1&\leq& \min(I(W;Y),I(W;Z));\label{eqn:DefOfMKGn1}\\
    R'_2&\leq& I(UW;Y);\label{eqn:DefOfMKGn2}\\
    R'_3&\leq& I(VW;Z);\label{eqn:DefOfMKGn3}\\
    R'_4&\leq&
I(U;Y|W)+I(V;Z|W)-I(U;V|W)\nonumber\\&&+\min(I(W;Y),I(W;Z)).\label{eqn:DefOfMKGn4}
  \end{eqnarray}
over auxiliary random variables satisfying the cardinality bounds
$|\mathcal{U}|\leq S_u$, $|\mathcal{V}|\leq S_v$ and
$|\mathcal{W}|\leq S_w$. Note that the region ${\mathcal{C}}_{M-I}^{S_u, S_v, S_w}$ specifies ${\mathcal{C}}_{M}^{S_u, S_v, S_w}$, since given any $p(u,v,w,x,y,z)=p(u,v,w,x)q(y,z|x)$ the
corresponding vector in ${\mathcal{C}}_{M-I}^{S_u, S_v, S_w}$ is providing
the values for the right hand side of the 4 inequalities that define
the region ${\mathcal{C}}_{M}^{S_u, S_v, S_w}$. Also note that $\mathcal{C}_{M-I}(q(y,
z|x))$ is defined as a subset of $\mathbb{R}^4$, and not
$\mathbb{R}_{+}^4$.

It is proved in Appendix \ref{sec:apndxNew2} that ${\mathcal{C}}_{M-I}^{S_u, S_v, |\mathcal{X}|+4}(q(y, z|x))$ is convex and closed for any $S_u$ and $S_v$. Thus, to prove equation (\ref{eqn:MP1}) it suffices to show that for any real $\lambda_1, \lambda_2, \lambda_3, \lambda_4$,
$$\max_{(R'_1,R'_2,R'_3,R'_4)\in {\mathcal{C}}_{M-I}^{S_u, S_v, |\mathcal{X}|+4}}\sum_{i=1:4}\lambda_iR'_i\leq
\max_{(R'_1,R'_2,R'_3,R'_4)\in {\mathcal{C}}_{M-I}^{|\mathcal{X}|, |\mathcal{X}|, |\mathcal{X}|+4}}\sum_{i=1:4}\lambda_iR'_i.$$
It suffices to prove this for the case of $\lambda_i\geq 0$ for $i=1:4$,
since if $\lambda_i$ is negative for some $i$, $R'_i$ can be made to converge to
$-\infty$ causing $\sum_{i=1}^4\lambda_iR'_i$ to converge to
$\infty$ on both sides. 

Take a point $(R'_1,R'_2,R'_3,R'_4)\in {\mathcal{C}}_{M-I}^{S_u, S_v, |\mathcal{X}|+4}$ that maximizes $\sum_{i=1:4}\lambda_iR'_i$. Corresponding to the point is a joint distribution $p(u,v,w,x)$ where $|U|\leq S_u$, $|V|\leq S_v$ and $|W|\leq |\mathcal{X}|+4$ and \begin{align*}\sum_{i=1:4}\lambda_iR'_i=&\lambda_1\min(I(W;Y),I(W;Z))+\lambda_2 I(UW;Y)+\lambda_3I(VW;Z)\\&+\lambda_4 \big(\min(I(W;Y),I(W;Z))+ I(U;Y|W)+I(V;Z|W)-I(U;V|W)\big).\end{align*} Let us fix $p(w,x)$. We would like to define $p(\widehat{u},\widehat{v}|w,x)$ such that $|\widehat{U}|\leq |X|$, $|\widehat{V}|\leq |X|$ achieving the same or larger weighted sum. Because we have fixed $p(w,x)$, the terms $I(W;Y)$ and $I(W;Z)$ are fixed. Since $I(UW;Y)=I(W;Y)+\sum_{w}p(w)I(U;Y|W=w)$, $I(VW;Z)=I(W;Z)+\sum_{w}p(w)I(V;Z|W=w)$ and $I(U;Y|W)+I(V;Z|W)-I(U;V|W)=\sum_{w}p(w)[I(U;Y|W=w)+I(V;Z|W=w)-I(U;V|W=w)]$, we can construct $p(\widehat{u},\widehat{v},x|w)$ for each $w$ individually. In other words, given the marginal distribution $p(x|w)$, we would like to construct $p(\widehat{u},\widehat{v},x|w)$ such that 
\begin{align*}&\lambda_2 I(U;Y|W=w)+\lambda_3I(V;Z|W=w)+\lambda_4 \big(I(U;Y|W=w)+I(V;Z|W=w)-I(U;V|W=w)\big)\leq\\&
\lambda_2 I(\widehat{U};Y|W=w)+\lambda_3I(\widehat{V};Z|W=w)+\lambda_4 \big(I(\widehat{U};Y|W=w)+I(\widehat{V};Z|W=w)-I(\widehat{U};\widehat{V}|W=w)\big).\end{align*}
When $\lambda_4>0$, after a normalization we get the problem studied in section \ref{Section:Description}. When $\lambda_4=0$, clearly $\widehat{U}=\widehat{V}=X$ works. This completes the proof. Thus, we have proved that for any arbitrary natural numbers $S_u, S_v, S_w$, one has ${\mathcal{C}}_{M}^{S_u, S_v, S_w}(q(y, z|x))\subset
{\mathcal{C}}_{M}^{|\mathcal{X}|, |\mathcal{X}|, |\mathcal{X}|+4}(q(y, z|x))={\mathscr{L}}(q(y, z|x))$.

We now complete the proof of the theorem. 
In Appendices \ref{sec:apndxL} and \ref{sec:apndxM}, we prove that
the closure of ${\mathcal{C}}_{M}(q(y, z|x))$ is equal to the
closure of $\bigcup_{S_u, S_v, S_w\geq 0}{\mathcal{C}}_{M}^{S_u,
S_v, S_w}(q(y, z|x))$, and that ${\mathscr{C}}(q(y, z|x))$ is equal
to ${\mathscr{L}}(q(y, z|x))$. Using the result that ${\mathcal{C}}_{M}^{S_u, S_v, S_w}(q(y, z|x))\subset
{\mathcal{C}}_{M}^{|\mathcal{X}|, |\mathcal{X}|, |\mathcal{X}|+4}(q(y, z|x))={\mathscr{L}}(q(y, z|x))$, we get that the closure of ${\mathcal{C}}_{M}(q(y, z|x))$ is equal to the
closure of ${\mathscr{L}}(q(y, z|x))$. Lastly note that $\mathscr{L}(q(y, z|x))$ is closed because of the cardinality constraints on its auxiliary random variables.\footnote{Since the ranges of all the involved
random variables are limited and the conditional mutual information
function is continuous, the set of admissible joint probability
distributions $p(u,v,w,x,y,z)$ where $I(UVW;YZ|X)=0$ and
$p(y,z|x)=q(y,z|x)$ will be a compact set (when viewed as a subset
of the ambient Euclidean space). The fact that mutual information function
is continuous implies that the Marton region defined by taking the union over random variables
$U,V,W,X,Y,Z$ satisfying the cardinality bounds is a compact set, and thus closed.}
\end{proof}

\begin{proof}[Proof of Theorem \ref{Thm:Thm2}] We construct a broadcast channel with binary input alphabet for which Marton's inner bound and the recent outer bound of Nair and El Gamal do not match.

We begin by proving that for any arbitrary binary input broadcast
channel $q(y,z|x)$ such that for all $y\in \mathcal{Y}$ and $z\in
\mathcal{Z}$, $q(Y=y|X=0)$, $q(Y=y|X=1)$, $q(Z=z|X=0)$ and
$q(Z=z|X=1)$ are non-zero, the following holds: 

\emph{Lemma:} If $\mathcal{C}_{M}(q(y,
z|x))=\mathcal{C}_{NE}(q(y, z|x))$, the maximum sum rate
$R_1+R_2$ over triples $(R_0, R_1, R_2)$ in Marton's inner bound
is equal to
\begin{eqnarray}&\max\bigg(\min_{\gamma\in
[0,1]}\big(\max_{\scriptsize{\begin{array}{l}
              p(wx)q(y,z|x)\\
              |\mathcal{W}|=2
            \end{array}}}
\gamma I(W;Y)+ (1-\gamma)I(W;Z)+\sum_{w}p(w)T(p(X=1|W=w))\big),
\nonumber
\\& \max_{\scriptsize{\begin{array}{l}
              p(u,v)p(x|uv)q(y,z|x)\\
              |\mathcal{U}|=|\mathcal{V}|=2, I(U;V)=0, H(X|UV)=0
            \end{array}}}
I(U;Y)+I(V;Z)\bigg),\label{eqnNW1:Thm3}\end{eqnarray} where $T(p)=
\max\big\{I(X;Y), I(X;Z)|P(X=1)=p\big\}$.

Before proceeding to prove the above lemma, note that if the expression of equation \ref{eqnNW1:Thm3} turns out to be strictly less than the maximum of
the sum rate $R_1+R_2$ over triples $(R_0, R_1, R_2)$ in
$\mathcal{C}_{NE}(q(y, z|x))$ (which is given in \cite{NairZizhou}),
it will serve as an evidence for $\mathcal{C}_{M}(q(y, z|x))\neq
\mathcal{C}_{NE}(q(y, z|x))$. The maximum of the sum rate $R_1+R_2$
over triples $(R_0, R_1, R_2)$ in $\mathcal{C}_{NE}(q(y, z|x))$ is
known to be \cite{NairZizhou}
$$\max_{\scriptsize{\begin{array}{l}
              p(u,v,x)q(y,z|x)\\
            \end{array}}}
\min \big(I(U;Y)+I(V;Z), I(U;Y)+I(V;Z|U), I(V;Z)+I(U;Y|V)\big),$$
which can be written as (see Bound 4 in \cite{NairZizhou})
$$\max_{\scriptsize{\begin{array}{l}
              p(u,v,x)q(y,z|x)\\
              |\mathcal{U}|=|\mathcal{V}|=3, I(U;V|X)=0
            \end{array}}}
\min \big(I(U;Y)+I(V;Z), I(U;Y)+I(X;Z|U), I(V;Z)+I(X;Y|V)\big).$$
The constraint $I(U;V|X)=0$ is imposed because the outer bound depends only on the marginals $p(u,x)$ and $p(v,x)$.
There are examples for which the expression of equation
\ref{eqnNW1:Thm3} turns out to be strictly less than the maximum of
the sum rate $R_1+R_2$ over triples $(R_0, R_1, R_2)$ in
$\mathcal{C}_{NE}(q(y, z|x))$. For instance given any two positive
reals $\alpha$ and $\beta$ in the interval $(0,1)$, consider the
broadcast channel for which
$|\mathcal{X}|=|\mathcal{Y}|=|\mathcal{Z}|=2$, $p(Y=0|X=0)=\alpha$,
$p(Y=0|X=1)=\beta$, $p(Z=0|X=0)=1-\beta$, $p(Z=0|X=1)=1-\alpha$.
Assuming $\alpha=0.01$, Figure \ref{fig:AlphaOnePercent} plots
maximum of the sum rate for $C_{NE}(q(y,z|x))$, and maximum of the
sum rate for $C_{M}(q(y,z|x))$ (assuming that $C_{NE}(q(y,z|x)) =
C_{M}(q(y,z|x))$) as a function of $\beta$. Where the two curves do
not match, Nair and El Gamal's outer bound and Marton's inner bound
can not be equal for the corresponding broadcast channel.
\begin{figure}
\centering
\includegraphics[width=110mm]{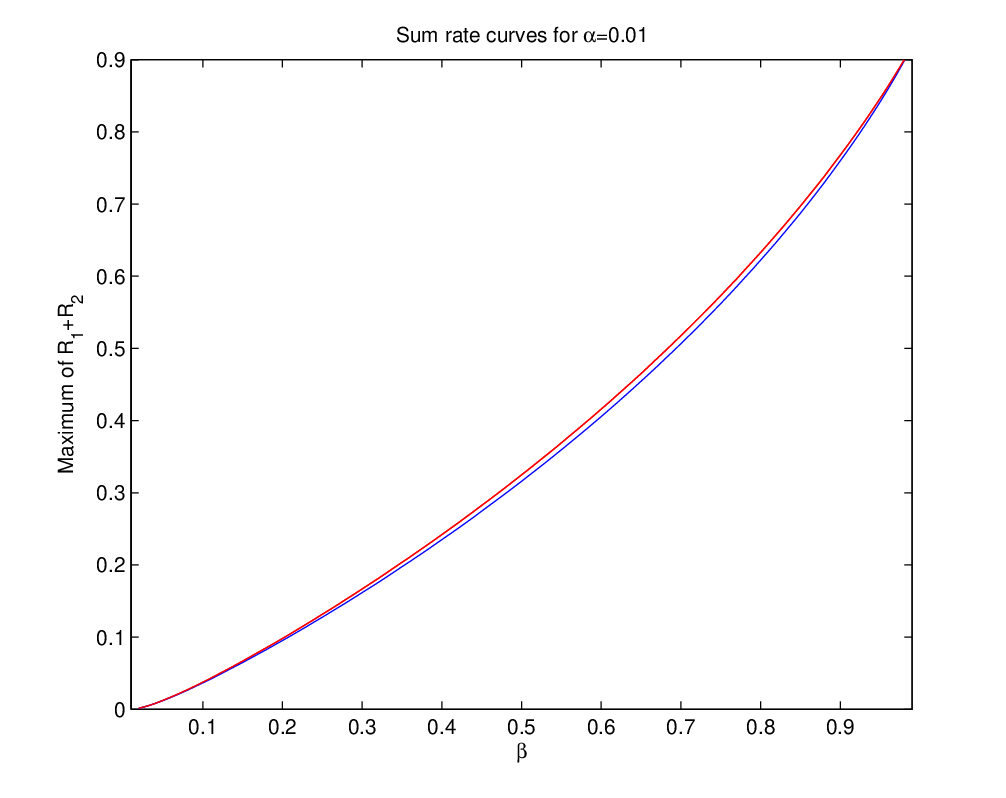}
\caption{Red curve (top curve): sum rate for $C_{NE}(q(y,z|x))$;
Blue curve (bottom curve): sum rate for $C_{M}(q(y,z|x))$ assuming
that $C_{NE}(q(y,z|x)) =
C_{M}(q(y,z|x))$.}\label{fig:AlphaOnePercent}
\end{figure}

\emph{Proof of the lemma:} The maximum of the sum rate $R_1+R_2$ over triples $(R_0, R_1, R_2)$ in
$\mathcal{C}_{M}(q(y, z|x))$ is equal to
\begin{equation}\label{eqn:PThm41}\max_{\scriptsize{\begin{array}{l}
              p(u,v,w,x)q(y,z|x)\\
              |\mathcal{U}|=2, |\mathcal{V}|=2\\
              H(X|UVW)=0
            \end{array}}}
I(U;Y|W)+I(V;Z|W)-I(U;V|W)+\min(I(W;Y), I(W;Z)).\end{equation} The
proof consists of two parts: first we show that the above expression
is equal to the following expression:
\begin{equation}\label{eqn:PThm42}\max\bigg(\max_{\scriptsize{\begin{array}{l}
              p(wx)q(y,z|x)
            \end{array}}}
\min\big(I(W;Y), I(W;Z)\big)+\sum_{w}p(w)T(p(X=1|W=w)),
\end{equation}$$ \max_{\scriptsize{\begin{array}{l}
              p(u,v)p(x|uv)q(y,z|x)\\
              |\mathcal{U}|=|\mathcal{V}|=2, I(U;V)=0, H(X|UV)=0
            \end{array}}}
I(U;Y)+I(V;Z)\bigg).$$ Next, we show that the expression of equation
\ref{eqn:PThm42} is equal to the the expression given in the lemma.

The expression of equation \ref{eqn:PThm41} is greater than or equal
to the expression of equation \ref{eqn:PThm42}.\footnote{Consider
the following special cases: 1) given $W=w$, let $(U,V)=(X,
constant)$ if $I(X;Y|W=w)\geq I(X;Z|W=w)$, and $(U,V)=(constant, X)$
otherwise. This would produce the first part of the expression given
in the lemma. 2) Assume that $W$ is constant, and $U$
is independent of $V$. This would produce the second part of the
expression given in the lemma.} For the first part of
the proof we thus need to prove that the expression of equation
\ref{eqn:PThm41} is less than or equal to the expression of equation
\ref{eqn:PThm42}. Take the joint distribution $p(u,v,w,x)$ that
maximizes the expression of equation \ref{eqn:PThm41}. Let
$\widetilde{U}=(U,W)$ and $\widetilde{V}=(V,W)$. The maximum of the sum
rate $R_1+R_2$ over triples $(R_0, R_1, R_2)$ in
$\mathcal{C}_{NE}(q(y, z|x))$ is greater than or equal to $\min
\big(I(\widetilde{U};Y)+I(\widetilde{V};Z),
I(\widetilde{U};Y)+I(\widetilde{V};Z|\widetilde{U}),
I(\widetilde{V};Z)+I(\widetilde{U};Y|\widetilde{V})\big)$ (see Bound
3 in \cite{NairZizhou}). Since $\mathcal{C}_{NE}(q(y,
z|x))=\mathcal{C}_{M}(q(y, z|x))$, we must have: $$\min
\big(I(UW;Y)+I(VW;Z), I(UW;Y)+I(VW;Z|UW),
I(UW;Z)+I(UW;Y|VW)\big)\leq $$$$
I(U;Y|W)+I(V;Z|W)-I(U;V|W)+\min(I(W;Y), I(W;Z)).$$ Or alternatively
$$\min \bigg(\max(I(W;Y),I(W;Z))+I(U;V|W),$$$$ I(W;Y)-\min(I(W;Y),
I(W;Z))+I(U;V|WZ),$$$$ I(W;Z)-\min(I(W;Y),
I(W;Z))+I(U;V|WY)\bigg)\leq 0.$$ Since each expression is also
greater than or equal to zero, at least one of the three terms must be equal
to zero. Therefore at least one of the following must hold:
\begin{enumerate}
  \item $I(W;Y)=I(W;Z)=0$ and $I(U;V|W)=0$,
  \item $I(U;V|WY)=0$,
  \item $I(U;V|WZ)=0$.
\end{enumerate}
If (1) holds, $I(U;Y|W)+I(V;Z|W)-I(U;V|W)+\min(I(W;Y), I(W;Z))$
equals $I(U;Y|W)+I(V;Z|W)$. Suppose $\max_{w:
p(w)>0}I(U;Y|W=w)+I(V;Z|W=w)$ occurs at some $w^*$. Clearly
$I(U;Y|W)+I(V;Z|W)\leq I(U;Y|W=w^*)+I(V;Z|W=w^*)$. Let $\widehat{U},
\widehat{V},\widehat{X}, \widehat{Y}$ and $\widehat{Z}$ be
distributed according to $p(u,v,x,y,z|w^*)$.
$I(\widehat{U};\widehat{V})=I(U;V|W=w^*)=0$. Therefore
$I(U;Y|W)+I(V;Z|W)-I(U;V|W)+\min(I(W;Y), I(W;Z))$ is less than or
equal to $$\max_{\scriptsize{\begin{array}{l}
              p(u,v)p(x|uv)q(y,z|x)\\
              |\mathcal{U}|=|\mathcal{V}|=2, I(U;V)=0, H(X|UV)=0
            \end{array}}}
I(U;Y)+I(V;Z).$$

Next assume (2) or (3) holds, i.e. $I(U;V|WY)=0$ or $I(U;V|WZ)=0$.
We show in Appendix \ref{sec:apndxVI} that for any value of $w$
where $p(w)>0$, either $I(U;V|W=w, Y)=0$ or $I(U;V|W=w, Z)=0$ imply
that $I(U;Y|W=w)+I(V;Z|W=w)-I(U;V|W=w)\leq T(p(X=1|W=w))$. Therefore
$I(U;Y|W)+I(V;Z|W)-I(U;V|W)+\min(I(W;Y), I(W;Z))\leq \min(I(W;Y),
I(W;Z))+\sum_{w}p(w)T(p(X=1|W=w))$. This in turn implies that
$I(U;Y|W)+I(V;Z|W)-I(U;V|W)+\min(I(W;Y), I(W;Z))$ is less than or
equal to $$\max_{\scriptsize{\begin{array}{l}
              p(w,x)q(y,z|x)\end{array}}}
\min(I(W;Y), I(W;Z))+\sum_{w}p(w)T(p(X=1|W=w)).$$ This completes the
first part of the proof.

Next, we would like to show that the expression of equation
\ref{eqn:PThm42} is equal to the the expression given in the lemma. In order to show this, note that (see Observation 1 of \cite{GohariElGamalAnantharam})
\begin{equation}\label{eqn:PThm43}\max_{\scriptsize{\begin{array}{l}
              p(w,x)q(y,z|x)\end{array}}}
\min(I(W;Y), I(W;Z))+\sum_{w}p(w)T(p(X=1|W=w))\end{equation} is
equal to
\begin{equation}\label{eqn:PThm44}\min_{\gamma\in
[0,1]}\big(\max_{\scriptsize{\begin{array}{l}
              p(wx)q(y,z|x)\\
              |\mathcal{W}|=2
            \end{array}}}
\gamma I(W;Y)+
(1-\gamma)I(W;Z)+\sum_{w}p(w)T(p(X=1|W=w))\big).\end{equation}

The expression given in equation \ref{eqn:PThm43} can be written as
$$\max_{\scriptsize{\begin{array}{l}
              p(w,x)q(y,z|x)\end{array}}}
\min\big(I(W;Y)+\sum_{w}p(w)T(p(X=1|W=w)),
I(W;Z)+\sum_{w}p(w)T(p(X=1|W=w))\big).$$ This expression can be rewritten as
$$\min_{\gamma\in
[0,1]}\big(\max_{\scriptsize{\begin{array}{l}
              p(wx)q(y,z|x)\\
            \end{array}}}
\gamma I(W;Y)+ (1-\gamma)I(W;Z)+\sum_{w}p(w)T(p(X=1|W=w))\big).$$ It
remains to prove the cardinality bound of two on $W$. This is done
using the strengthened Carath\'{e}odory theorem of Fenchel. Take an arbitrary $p(w,x)q(y,z|x)$. The vector
$w\rightarrow p(W=w)$ belongs to the set of vectors $w\rightarrow
p(\widetilde{W}=w)$ satisfying the constraints
$\sum_{w}p(\widetilde{W}=w)=1$, $p(\widetilde{W}=w)\geq 0$ and
$p(X=1)=\sum_wp(X=1|W=w)p(\widetilde{W}=w)$. The first two
constraints ensure that $w\rightarrow p(\widetilde{W}=w)$
corresponds to a probability distribution, and the third constraint
ensures that one can define a random variable $\widetilde{W}$, jointly
distributed with $X$, $Y$ and $Z$ according to $p(\widetilde{w},
x)q(y,z|x)$ and further satisfying
$p(X=x|\widetilde{W}=w)=p(X=x|W=w)$. Since $w\rightarrow p(W=w)$
belongs to the above set, it can be written as the convex
combination of some of the extreme points of this set. The
expression $\sum_w[-(1-\gamma)H(Z|W=w)-\gamma H(Y|W=w)+
T(p(X=1|W=w))]p(\widetilde{W}=w)$ is linear in $p(\widetilde{W}=w)$,
therefore this expression for $w\rightarrow p(W=w)$ is less than or
equal to the corresponding expression for at least one of these
extreme points. On the other hand, every extreme point of the set of
vectors $w\rightarrow p(\widetilde{W}=w)$ satisfying the constraints
$\sum_{w}p(\widetilde{W}=w)=1$, $p(\widetilde{W}=w)\geq 0$ and
$p(X=1)=\sum_wp(X=1|W=w)p(\widetilde{W}=w)$ satisfies the property
that $p(\widetilde{W}=w)\neq 0$ for at most two values of $w \in
\mathcal{W}$. Thus a cardinality bound of two is established.
\end{proof}

\begin{proof}[Proof of Lemma \ref{Lemma3}]
The equation $H(\widehat{X})=H(X)+\epsilon
H_L(X)-\mathbb{E}\big[r\big(\epsilon\cdot \mathbb{E}[L|X]\big)\big]$
where $r(x)=(1+x)\log(1+x)$ is true because:
\begin{eqnarray*}&H(\widehat{X})=-\sum_{\widehat{x}}p_{\epsilon}(\widehat{x})\log
p_{\epsilon}(\widehat{x})\\&=-\sum_{\widehat{x}}p_{0}(\widehat{x})\big(1+\epsilon\cdot
\mathbb{E}[L|X=\widehat{x}]\big)\cdot \log\bigg(
p_{0}(\widehat{x})\cdot\big(1+\epsilon\cdot
\mathbb{E}[L|X=\widehat{x}]\big)\bigg)\\&=-\sum_{\widehat{x}}p_{0}(\widehat{x})\big(1+\epsilon\cdot
\mathbb{E}[L|X=\widehat{x}] \big)\cdot \bigg[\log\bigg(
p_{0}(\widehat{x})\bigg)+\log\bigg(1+\epsilon\cdot
\mathbb{E}[L|X=\widehat{x}]\bigg)\bigg]\\&=H(X)-\epsilon\sum_{\widehat{x}}p_{0}(\widehat{x})\mathbb{E}[L|X=\widehat{x}]\log\bigg(
p_{0}(\widehat{x})\bigg)-\\&\sum_{\widehat{x}}p_{0}(\widehat{x})\big(1+\epsilon\cdot
\mathbb{E}[L|X=\widehat{x}]\big)\cdot \log\bigg(1+\epsilon\cdot
\mathbb{E}[L|X=\widehat{x}]\bigg)\\&=H(X)+\epsilon
H_L(X)-\mathbb{E}\big[r\big(\epsilon\cdot
\mathbb{E}[L|X]\big)\big].\end{eqnarray*}
\begin{figure}
\centering
\includegraphics[width=100mm]{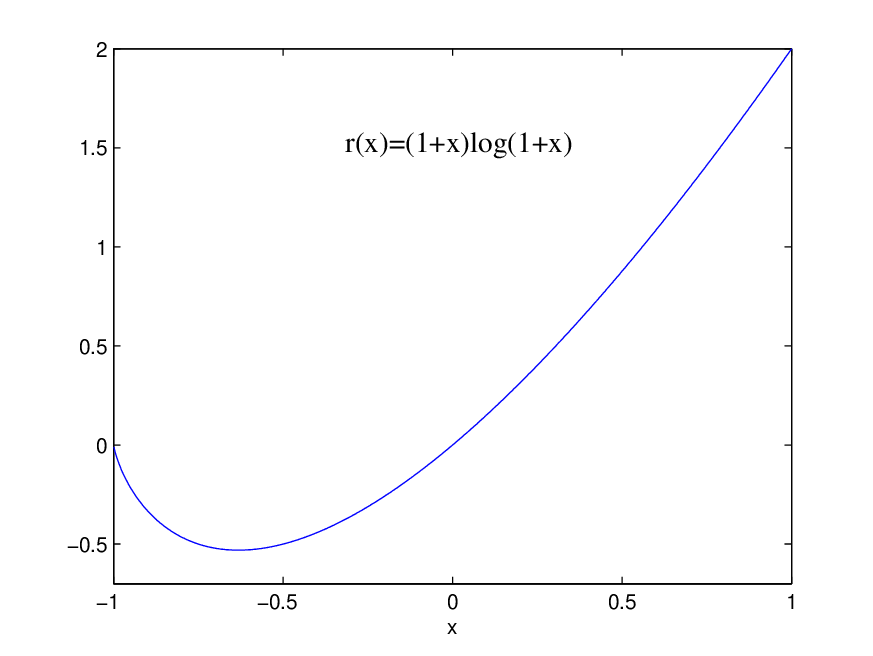}
\caption{Plot of the convex function $r(x)=(1+x)\log(1+x)$ over the
interval $[-1,1]$. Note that $r(0)=0$, $\frac{\partial}{\partial x}
r(x)=\log(1+x)+\log(e)$ and  $\frac{\partial^2}{\partial x^2}
r(x)=\frac{\log(e)}{1+x}>0$. }\label{fig:PlotOfRx}
\end{figure}
Next, note that $r(0)=0$, $\frac{\partial}{\partial x}
r(x)=\log(1+x)+\log(e)$ and  $\frac{\partial^2}{\partial x^2}
r(x)=\frac{\log(e)}{1+x}$. We have:
$$\frac{\partial}{\partial \epsilon}H(\widehat{X})= H_L(X) -
\mathbb{E}\big[\mathbb{E}[L|X]\{\log(1+\epsilon\cdot
\mathbb{E}[L|X])+\log e\}\big]= H_L(X) -
\mathbb{E}\big[\mathbb{E}[L|X]\log(1+\epsilon\cdot
\mathbb{E}[L|X])\big],$$ where at $\epsilon=0$ is equal to $H_L(X)$.

Next, we have:
\begin{eqnarray*}&\frac{\partial^2}{\partial \epsilon^2}H(\widehat{X})=
-\frac{\partial}{\partial \epsilon}
\mathbb{E}\big[\mathbb{E}[L|X]\log(1+\epsilon\cdot
\mathbb{E}[L|X])\big]\\&-
\mathbb{E}\big[\mathbb{E}[L|X]\frac{\mathbb{E}[L|X]}{1+\epsilon\cdot
\mathbb{E}[L|X]}\log e\big]=-\log e\cdot
\mathbb{E}\big[\frac{\mathbb{E}[L|X]^2}{1+\epsilon\cdot
\mathbb{E}[L|X]}\big]\end{eqnarray*}

On the other hand,
\begin{eqnarray*}&I(\epsilon)=\sum_{x}\bigg(\frac{\partial}{\partial \epsilon}\log_e(p_{\epsilon}(\widehat{X}=x))\bigg)^2p_{\epsilon}(\widehat{X}=x)
=\\&\sum_{x}\bigg(\frac{\partial}{\partial
\epsilon}\log_e\bigg(p_0(X=x)\cdot \big(1+\epsilon\cdot
\mathbb{E}[L|X=x]\big)\bigg)\bigg)^2p_0(X=x)\cdot
\big(1+\epsilon\cdot \mathbb{E}[L|X=x]\big)
=\\&\sum_{x}\bigg(\frac{\partial}{\partial \epsilon}\log_e
\big(1+\epsilon\cdot \mathbb{E}[L|X=x]\big)\bigg)^2p_0(X=x)\cdot
\big(1+\epsilon\cdot \mathbb{E}[L|X=x]\big)
=\\&\sum_{x}\bigg(\frac{\mathbb{E}[L|X=x]}{1+\epsilon\cdot
\mathbb{E}[L|X=x]}\bigg)^2p_0(X=x)\cdot \big(1+\epsilon\cdot
\mathbb{E}[L|X=x]\big)=
\\&\sum_{x}\frac{\mathbb{E}[L|X=x]^2}{1+\epsilon\cdot
\mathbb{E}[L|X=x]}p_0(X=x)=\mathbb{E}\big[\frac{\mathbb{E}[L|X]^2}{1+\epsilon\cdot
\mathbb{E}[L|X]}\big].\end{eqnarray*}
\end{proof}
\appendices
\section{}\label{sec:apndxNew2}
In this appendix we show that ${\mathcal{C}}_{M-I}^{S_u, S_v, |\mathcal{X}|+4}(q(y, z|x))$ is convex and closed for any $S_u$ and $S_v$.
We begin by proving that the region ${\mathcal{C}}_{M-I}^{S_u, S_v,
S_w}$ is closed. Since the ranges of all the involved
random variables are limited and the conditional mutual information
function is continuous, the set of admissible joint probability
distributions $p(u,v,w,x,y,z)$ where $I(UVW;YZ|X)=0$ and
$p(y,z|x)=q(y,z|x)$ will be a compact set (when viewed as a subset
of the ambient Euclidean space). The fact that mutual information function
is continuous implies that the union over random variables
$U,V,W,X,Y,Z$ satisfying the cardinality bounds, having the joint
distribution $p(u,v,w,x,y,z)=p(u,v,w,x)q(y,z|x)$, of the region defined by equations (\ref{eqn:DefOfMKGn1}-\ref{eqn:DefOfMKGn4}) is compact, and thus closed.

Next we prove that ${\mathcal{C}}_{M-I}^{S_u, S_v, |\mathcal{X}|+4}(q(y, z|x))$ is convex. Since
${\mathcal{C}}_{M-I}^{S_u, S_v, S_w}(q(y, z|x))$ is a subset of
${\mathcal{C}}_{M-I}^{S_u, S_v, |\mathcal{X}|+4}(q(y, z|x))$ as mentioned in step 1 in the proof of Theorem \ref{Thm:Thm1}, it
suffices to show that $\bigcup_{S_w\geq 0}{\mathcal{C}}_{M-I}^{S_u,
S_v, S_w}(q(y, z|x))$ is convex. Take two arbitrary points $(R_1,
R_2, ..., R_4)$ and $(\widetilde{R_1}, \widetilde{R_2}, ...,
\widetilde{R_4})$ in $\bigcup_{S_w\geq 0}{\mathcal{C}}_{M-I}^{S_u,
S_v, S_w}(q(y, z|x))$. Corresponding to $(R_1, ..., R_4)$ and
$(\widetilde{R_1}, ..., \widetilde{R_4})$ are joint distributions
$p_0(u,v,w,x,y,z)=p_0(u,v,w,x)q(y,z|x)$ on $U,V,W,X,Y,Z$, and
$p_0(\widetilde{u},\widetilde{v},\widetilde{w},\widetilde{x},\widetilde{y},\widetilde{z})=p_0(\widetilde{u},\widetilde{v},\widetilde{w},\widetilde{x})q(\widetilde{y},\widetilde{z}|\widetilde{x})$
on
$\widetilde{U},\widetilde{V},\widetilde{W},\widetilde{X},\widetilde{Y},\widetilde{Z}$,
 where $|\mathcal{U}|=|\mathcal{\widetilde{U}}|=S_u$,
$|\mathcal{V}|=|\mathcal{\widetilde{V}}|=S_v$, and furthermore the
following equations are satisfied: $R_1\leq \min(I(W;Y),I(W;Z))$, $R_2\leq I(UW;Y)$, ..., $\widetilde{R_1}\leq
\min(I(\widetilde{W};\widetilde{Y}),
I(\widetilde{W};\widetilde{Z}))$, $\widetilde{R_2}\leq
I(\widetilde{U}\widetilde{W};\widetilde{Y})$, ... etc. 

Without loss of generality we can assume that
($\widetilde{U},\widetilde{V},\widetilde{W},\widetilde{X},\widetilde{Y},\widetilde{Z}$)
and $(U,V,W,X,Y,Z)$ are independent. Let $Q$ be a uniform binary
random variable independent of all previously defined random
variables. Let
$(\widehat{U},\widehat{V},\widehat{W},\widehat{X},\widehat{Y},\widehat{Z})$
be equal to $(U,V,WQ,X,Y,Z)$ when $Q=0$, and equal to
($\widetilde{U},\widetilde{V},\widetilde{W}Q,\widetilde{X},\widetilde{Y},\widetilde{Z}$)
when $Q=1$. One can verify that
$p(\widehat{Y}=y,\widehat{Z}=z|\widehat{X}=x)=q(\widehat{Y}=y,\widehat{Z}=z|\widehat{X}=x)$,
$I(\widehat{U}\widehat{V}\widehat{W};\widehat{Y}\widehat{Z}|\widehat{X})=0$,
and furthermore
\begin{eqnarray*}&I(\widehat{W};\widehat{Y})\geq
\frac{1}{2}I(W;Y)+\frac{1}{2}I(\widetilde{W};\widetilde{Y})\\&
I(\widehat{W};\widehat{Z})\geq
\frac{1}{2}I(W;Z)+\frac{1}{2}I(\widetilde{W};\widetilde{Z})\\&
I(\widehat{U}\widehat{W};\widehat{Y})\geq
\frac{1}{2}I(UW;Y)+\frac{1}{2}I(\widetilde{U}\widetilde{W};\widetilde{Y})\\&...
\end{eqnarray*}
Hence $(\frac{1}{2}R_1+\frac{1}{2}\widetilde{R_1},
\frac{1}{2}R_2+\frac{1}{2}\widetilde{R_2}, ...,
\frac{1}{2}R_4+\frac{1}{2}\widetilde{R_4})$ belongs to
$\bigcup_{S_w\geq 0}{\mathcal{C}}_{M-I}^{S_u, S_v, S_w}(q(y, z|x))$.
Thus $\\ \bigcup_{S_w\geq 0}{\mathcal{C}}_{M-I}^{S_u, S_v, S_w}(q(y,
z|x))={\mathcal{C}}_{M-I}^{S_u, S_v, |\mathcal{X}|+4}(q(y, z|x))$ is
convex.

\section{}\label{sec:apndxL}
In this appendix, we prove that the closure of
${\mathcal{C}}_{M}(q(y, z|x))$ is equal to the closure of
$\\\bigcup_{S_u, S_v, S_w\geq 0}{\mathcal{C}}_{M}^{S_u, S_v,
S_w}(q(y, z|x))$. In order to show this it suffices to show that any
triple $(R_0, R_1, R_2)$ in ${\mathcal{C}}_{M}(q(y, z|x))$ is a
limit point of $\bigcup_{S_u, S_v, S_w\geq 0}{\mathcal{C}}_{M}^{S_u,
S_v, S_w}(q(y, z|x))$. Since $(R_0, R_1, R_2)$ is in
${\mathcal{C}}_{M}(q(y, z|x))$, random variables $U, V, W, X, Y$ and
$Z$ for which equations \ref{eqn:DefOfMKG1}, \ref{eqn:DefOfMKG2},
\ref{eqn:DefOfMKG3} and \ref{eqn:DefOfMKG4} are satisfied exist.
First assume $U, V, W$ are discrete random variables taking values
in $\{1,2,3,...\}$. For any integer $m$, let $U_m, V_m$ and $W_m$ be
truncated versions of $U, V$ and $W$ defined on $\{1,2,3,...,m \}$
as follows: $U_m, V_m$ and $W_m$ are jointly distributed according
to $p(U_m=u, V_m=v, W_m=w)=\frac{p(U=u, V=v, W=w)}{p(U\leq m, V\leq
m, W\leq m)}$ for every $u$, $v$ and $w$ less than or equal to $m$.
Further assume that $X_m$, $Y_m$ and $Z_m$ are random variables
defined on $\mathcal{X}$, $\mathcal{Y}$ and $\mathcal{Z}$ where
$p(Y_m=y,Z_m=z,X_m=x|U_m=u,V_m=v,W_m=w)=p(Y=y,Z=z,X=x|U=u,V=v,W=w)$
for every $u$, $v$ and $w$ less than or equal to $m$, and for every
$x$, $y$ and $z$. Note that the joint distribution of $U_m, V_m,
W_m, X_m, Y_m$ and $Z_m$ converges to that of $U, V, W, X, Y$ and
$Z$ as $m\rightarrow \infty$. Therefore the mutual information terms
$I(W_m;Y_m)$, $I(W_m; Z_m)$, $I(W_mU_m; Y_m)$, ... (that define a
region in ${\mathcal{C}}_{M}^{m, m, m}(q(y, z|x))$) converge to the
corresponding terms $I(W;Y)$, $I(W; Z)$, $I(WU; Y)$, ... Therefore
$(R_0, R_1, R_2)$ is a limit point of $\bigcup_{S_u, S_v, S_w\geq
0}{\mathcal{C}}_{M}^{S_u, S_v, S_w}(q(y, z|x))$.

Next assume that some of the random variables $U$, $V$ and $W$ are
continuous. Given any positive $q$, one can quantize the continuous
random variables to a precision $q$, and get discrete random
variables $U_q$, $V_q$ and $W_q$. We have already established that
any point in the Marton's inner bound region corresponding to $U_q,
V_q, W_q, X, Y, Z$ is a limit point of $\bigcup_{S_u, S_v, S_w\geq
0}{\mathcal{C}}_{M}^{S_u, S_v, S_w}(q(y, z|x))$. The joint
distribution of $U_q, V_q, W_q, X, Y, Z$ converges to that of $U, V,
W, X, Y, Z$ as $q$ converges to zero. Therefore the corresponding
mutual information terms $I(W_q;Y_q)$, $I(W_q; Z_q)$, $I(W_qU_q;
Y_q)$, ... (that define a region in ${\mathcal{C}}_{M}(q(y, z|x))$)
converge to the corresponding terms $I(W;Y)$, $I(W; Z)$, $I(WU;
Y)$,.... Therefore $(R_0, R_1, R_2)$ is a limit point of
$\bigcup_{S_u, S_v, S_w\geq 0}{\mathcal{C}}_{M}^{S_u, S_v, S_w}(q(y,
z|x))$.

\section{}\label{sec:apndxM}
In this appendix, we prove that ${\mathscr{C}}(q(y, z|x))$ is equal
to ${\mathscr{L}}(q(y, z|x))$. Clearly ${\mathscr{C}}(q(y,
z|x))\subset {\mathscr{L}}(q(y, z|x))$. Therefore we need to show
that ${\mathscr{L}}(q(y, z|x))\subset {\mathscr{C}}(q(y, z|x))$.

We need two definitions before proceeding. Let ${\mathscr{L}}'(q(y, z|x))$ be a subset of $\mathbb{R}^6$ defined
as the union of
\begin{eqnarray*}&\Delta\big(\big\{\big(I(W;Y), I(W;Z), I(UW;Y),
I(VW;Z),\\&I(U;Y|W)+I(V;Z|W)-I(U;V|W)+I(W;Y),\\&
I(U;Y|W)+I(V;Z|W)-I(U;V|W)+I(W;Z)\big)\big\}\big),\end{eqnarray*}
over random variables $U,V,W,X,Y,Z$, having the joint distribution
$p(u,v,w,x,y,z)=p(u,v,w,x)q(y,z|x)$ and satisfying the cardinality constraints $|\mathcal{U}|\leq |\mathcal{X}|$, $|\mathcal{V}|\leq |\mathcal{X}|$ and $|\mathcal{W}|\leq |\mathcal{X}|+4$.  ${\mathscr{C}}'(q(y, z|x))$ is defined similarly, except that the additional constraint $H(X|UVW)=0$ is imposed on the auxiliary random variables. Note that the region
${\mathscr{L}}'(q(y, z|x))$ specifies ${\mathscr{L}}(q(y, z|x))$, since given any $p(u,v,w,x,y,z)=p(u,v,w,x)q(y,z|x)$ the
corresponding vector in ${\mathscr{L}}'(q(y, z|x))$ is providing
the values for the right hand side of the 6 inequalities that define
the region ${\mathscr{L}}(q(y, z|x))$. Similarly ${\mathscr{C}}'(q(y, z|x))$ specifies ${\mathscr{C}}(q(y, z|x))$.

Instead of showing that ${\mathscr{L}}(q(y, z|x))\subset {\mathscr{C}}(q(y, z|x))$, it suffices to show that ${\mathscr{L}}'(q(y, z|x))\subset
{\mathscr{C}}'(q(y, z|x))$.\footnote{This is true because
$(R_{0}, R_{1}, R_{2})$ being in ${\mathscr{L}}(q(y, z|x))$ implies
that $(R_{0}, R_{0}, R_{0}+R_{1}, R_{0}+R_{2},
R_{0}+R_{1}+R_{2},R_{0}+R_{1}+R_{2})$ is in ${\mathscr{L}}'(q(y,
z|x))$. If ${\mathscr{L}}'(q(y, z|x))(q(y, z|x))$ is a subset of
${\mathscr{C}}'(q(y, z|x))$, the latter point would belong to
${\mathscr{C}}'(q(y, z|x))$. Therefore $(R_{0}, R_{1}, R_{2})$
belongs to $\mathscr{C}(q(y, z|x))$.} It suffices to prove that
${\mathscr{C}}'(q(y, z|x))$ is convex, and that for any
$\lambda_1$, $\lambda_2$, ..., $\lambda_6$, the maximum of
$\sum_{i=1}^6\lambda_iR_i$ over triples $(R_1, R_2, ..., R_6)$ in
$\mathscr{L}'(q(y, z|x))$, is less than or equal to the maximum
of $\sum_{i=1}^6\lambda_iR_i$ over triples $(R_1, R_2, ..., R_6)$ in
$\mathscr{C}'(q(y, z|x))$.

In order to show that ${\mathscr{C}}'(q(y, z|x))$ is convex, we
take two arbitrary points in ${\mathscr{C}}'(q(y, z|x))$.
Corresponding to them are joint distributions $p(u_1, v_1, w_1, x_1,
y_1, z_1)$ and $p(u_2, v_2, w_2, x_2, y_2, z_2)$. Let $Q$ be a
uniform binary random variable independent of all previously defined
random variables, and let $U=U_Q$, $V=V_Q$, $W=(W_Q,Q)$, $X=X_Q$,
$Y=Y_Q$ and $Z=Z_Q$. Clearly $H(X|UVW)=0$, and furthermore
$I(W;Y)\geq \frac{1}{2}\big(I(W_1;Y_1)+I(W_2;Y_2)\big)$, $I(W;Z)\geq
\frac{1}{2}\big(I(W_1;Z_1)+I(W_2;Z_2)\big)$, ... etc. Random variable
$W$ is not defined on an alphabet set of size
$|\mathcal{X}|+4$. However, one can reduce the cardinality of $W$
using the Carath\'{e}odory theorem by fixing $p(u,v,x,y,z|w)$ and changing the
marginal distribution of $W$ in a way that at most $|\mathcal{X}|+4$
elements get non-zero probability assigned to them. Since we have
preserved $p(u,v,x,y,z|w)$ throughout the process, $p(x|u,v,w)$ will
continue to belong to the set $\{0,1\}$ after reducing the cardinality
of $W$.

Next, we need to show that for any $\lambda_1$, $\lambda_2$, ...,
$\lambda_6$, the maximum of $\sum_{i=1}^6\lambda_iR_i$ over triples
$(R_1, R_2, ..., R_6)$ in $\mathscr{L}'(q(y, z|x))$, is less than
or equal to the maximum of $\sum_{i=1}^6\lambda_iR_i$ over triples
$(R_1, R_2, ..., R_6)$ in $\mathscr{C}'(q(y, z|x))$. As discussed
in the proof of theorem \ref{Thm:Thm1}, without loss of generality
we can assume $\lambda_i$ is non-negative for $i=1,2,...,6$.

Take an arbitrary point $(R_1, R_2, ..., R_6)$ in
$\mathscr{L}'(q(y, z|x))$. By definition there exist random
variables $U, V, W, X, Y$ and $Z$ for which
\begin{eqnarray}&\sum_{i=1}^6\lambda_iR_i\leq \lambda_1\cdot
I(W;Y)+\lambda_2\cdot I(W;Z)+\lambda_3\cdot I(UW;Y)+\lambda_4\cdot
I(VW;Z)+\label{eqn:AppndxM1}\\& \lambda_5\cdot
\big(I(U;Y|W)+I(V;Z|W)-I(U;V|W)+I(W;Y)\big)+\nonumber\\&\lambda_6\cdot
\big(I(U;Y|W)+I(V;Z|W)-I(U;V|W)+I(W;Z)\big).\nonumber\end{eqnarray}
Fix $p(u,v,w)$. The right hand side of equation (\ref{eqn:AppndxM1})
would then be a convex function of $p(x|u,v,w)$.\footnote{This is
true because $I(W;Y)$ is convex in the conditional distribution
$p(y|w)$; similarly $I(U;Y|W=w)$ is convex for any fixed value of
$w$. The term $I(U;V|W)$ that appears with a negative sign is
constant since the joint distribution of $p(u,v,w)$ is fixed.}
Therefore its maximum occurs at the extreme points when
$p(x|u,v,w)\in \{0,1\}$ whenever $p(u,v,w)\neq 0$. Therefore random
variables $\widehat{U}, \widehat{V}, \widehat{W}, \widehat{X},
\widehat{Y}$, and $\widehat{Z}$ exist for which
\begin{eqnarray*}&\lambda_1\cdot
I(W;Y)+\lambda_2\cdot I(W;Z)+...+\lambda_6\cdot
\big(I(U;Y|W)+I(V;Z|W)-I(U;V|W)+I(W;Z)\big)\leq\\& \lambda_1\cdot
I(\widehat{W};\widehat{Y})+\lambda_2\cdot
I(\widehat{W};\widehat{Z})+...+\lambda_6\cdot
\big(I(\widehat{U};\widehat{Y}|\widehat{W})+I(\widehat{V};\widehat{Z}|\widehat{W})-I(\widehat{U};\widehat{V}|\widehat{W})+I(\widehat{W};\widehat{Z})\big)\end{eqnarray*}
and furthermore $p(\widehat{x}|\widehat{u},\widehat{v},\widehat{w})
\in \{0,1\}$ for all $\widehat{x}$, $\widehat{u}$, $\widehat{v}$ and
$\widehat{w}$ where $p(\widehat{u},\widehat{v},\widehat{w})>0$.

\section{}\label{sec:apndxVI}
In this appendix, we complete the proof of Theorem \ref{Thm:Thm2} by
showing that given any random variables $U, V, W, X, Y$ and $Z$
where $p(u,v,w,x,y,z)=p(u,v,w,x)q(y,z|x)$ holds, $U$ and $V$ are
binary, $H(X|UVW)$ is zero, the transition matrices $P_{Y|X}$ and
$P_{Z|X}$ have positive elements, and for any value of $w$ where
$p(w)>0$, either $I(U;V|W=w, Y)=0$ or $I(U;V|W=w, Z)=0$ holds, the
following inequality is true: $$I(U;Y|W=w)+I(V;Z|W=w)-I(U;V|W=w)\leq
T(p(X=1|W=w)).$$

We assume $I(U;V|W=w, Y)=0$ (the proof for the case $I(U;V|W=w,
Z)=0$ is similar). First consider the case in which the individual
capacity $C_{P_{Y|X}}$ is zero. We will then have $I(U;Y|W=w)=0$ and
$T(p(X=1|W=w))=I(X;Z|W=w)\geq I(V;Z|W=w)-I(U;V|W=w)$. Therefore the
inequality holds in this case. Assume therefore that $C_{P_{Y|X}}$is
non-zero.

It suffices to prove the following proposition:

\emph{Proposition:} For any random variables $U,V,X,Y$ and $Z$
satisfying
\begin{itemize}
  \item $UV\rightarrow X\rightarrow YZ$,
  \item $H(X|UV)=0$,
  \item $|\mathcal{U}|=|\mathcal{V}|=|\mathcal{X}|=2$,
  \item for all $y\in \mathcal{Y}$, $p(Y=y|X=0)$ and
$p(Y=y|X=1)$ are non-zero,
  \item $C_{P_{Y|X}}\neq 0$,
  \item $I(U;V|Y)=0$,
\end{itemize}
one of the following two cases must be true: (1) at least one of the
random variables $X$, $U$ or $V$ is constant, (2) Either $U=X$ or
$U=1-X$ or $V=X$ or $V=1-X$.

\emph{Proof:} Assume that neither (1) nor (2) holds. Since
$H(X|UV)=0$, there are $2^4$ possible descriptions for $p(x|uv)$,
some of which are ruled out because neither (1) nor (2) holds. In
the following we prove that $X=U\oplus V$ and $X=U\wedge V$ can not
hold. The proof for other cases is essentially the same.

Since $C_{P_{Y|X}}\neq 0$, we conclude that the transition matrix
$P_{Y|X}$ has linearly independent rows. This implies the existence
of $y_1, y_2 \in \mathcal{Y}$ for which $p(X=1|Y=y_1)\neq
p(X=1|Y=y_2)$.\footnote{If this were not the case we would have we have
$p(X=1|Y=y_1)= p(X=1|Y=y_2)$ for all $y_1, y_2 \in \mathcal{Y}$.
This would imply that $X$ and $Y$ are independent. Since $X$ is not
constant, independence of $X$ and $Y$ implies that
$P(Y=y|X=1)=p(Y=y|X=0)$ for all $y\in \mathcal{Y}$. Therefore the
transition matrix $P_{Y|X}$ has linearly dependent rows. Hence
$I(X;Y)=0$ for all $p(x)$. Therefore $C_{P_{Y|X}}=0$ which is a
contradiction.} Furthermore since $X$ is not constant, and
$p(Y=y_1|X=0), p(Y=y_1|X=1), p(Y=y_2|X=0)$, and $p(Y=y_2|X=1)$ are
all non-zero, both $p(X=1|Y=y_1)$ and $p(X=1|Y=y_2)$ are in the open
interval $(0,1)$. Note that $I(U;V|Y)=0$ implies that
$I(U;V|Y=y_1)=0$ and $I(U;V|Y=y_2)=0$.

Let $a_{i,j}=p(U=i,V=j)$ for $i,j \in \{0,1\}$. First assume that
$X=U\oplus V$. We have
\begin{itemize}
\item
$p(u=0,v=0|y=y_i)=\frac{a_{0,0}}{a_{0,0}+a_{1,1}}p(X=0|Y=y_i)$,
\item
$p(u=0,v=1|y=y_i)=\frac{a_{0,1}}{a_{0,1}+a_{1,0}}p(X=1|Y=y_i)$,
\item
$p(u=1,v=0|y=y_i)=\frac{a_{1,0}}{a_{0,1}+a_{1,0}}p(X=1|Y=y_i)$,
\item
$p(u=1,v=1|y=y_i)=\frac{a_{1,1}}{a_{0,0}+a_{1,1}}p(X=0|Y=y_i)$.
\end{itemize}
Therefore $I(U;V|Y=y_i)=0$ for $i=1,2$ implies that
$$p(u=1,v=1|y=y_i)\times p(u=0,v=0|y=y_i)= p(u=0,v=1|y=y_i)\times
p(u=1,v=0|y=y_i).$$ Therefore
$$\frac{a_{0,0}a_{1,1}}{(a_{0,0}+a_{1,1})^2}p(X=0|Y=y_i)^2=
\frac{a_{0,1}a_{1,0}}{(a_{0,1}+a_{1,0})^2}p(X=1|Y=y_i)^2,$$ or
alternatively
\begin{equation}\label{eqn:appendixIV1}\frac{\sqrt{a_{0,0}a_{1,1}}}{a_{0,0}+a_{1,1}}p(X=0|Y=y_i)=
\frac{\sqrt{a_{1,0}a_{0,1}}}{a_{1,0}+a_{0,1}}p(X=1|Y=y_i).\end{equation}
 Since $X$ is
not deterministic, $P(X=0)=a_{0,0}+a_{1,1}$ and
$P(X=1)=a_{1,0}+a_{0,1}$ are non-zero. Next, if either of $a_{0,0}$
or $a_{1,1}$ are zero, it implies that $a_{1,0}$ or $a_{0,1}$ is
zero. But this implies that either $U$ or $V$ is a constant random variable which is a contradiction. Hence
$\frac{\sqrt{a_{0,0}a_{1,1}}}{a_{0,0}+a_{1,1}}$ and
$\frac{\sqrt{a_{1,0}a_{0,1}}}{a_{1,0}+a_{0,1}}$ are non-zero. But
then equation \ref{eqn:appendixIV1} uniquely specifies
$p(X=1|Y=y_i)$, implying that $p(X=1|Y=y_1)=p(X=1|Y=y_2)$ which is
again a contradiction.

Next assume that $X=U\wedge V$. We have:
\begin{itemize}
\item
$p(u=0,v=0|y=y_i)=\frac{a_{0,0}}{a_{0,0}+a_{0,1}+a_{1,0}}p(X=0|Y=y_i)$,
\item
$p(u=0,v=1|y=y_i)=\frac{a_{0,1}}{a_{0,0}+a_{0,1}+a_{1,0}}p(X=0|Y=y_i)$,
\item
$p(u=1,v=0|y=y_i)=\frac{a_{1,0}}{a_{0,0}+a_{0,1}+a_{1,0}}p(X=0|Y=y_i)$,
\item
$p(u=1,v=1|y=y_i)=p(X=1|Y=y_i)$.
\end{itemize}
Note that $P(X=0)=a_{0,0}+a_{0,1}+a_{1,0}$ is non-zero. Independence
of $U$ and $V$ given $Y=y_i$ implies that
$$p(u=1,v=1|y=y_i)\times p(u=0,v=0|y=y_i)= p(u=0,v=1|y=y_i)\times
p(u=1,v=0|y=y_i).$$ Therefore
$$\frac{a_{0,0}}{a_{0,0}+a_{0,1}+a_{1,0}}p(X=0|Y=y_i)p(X=1|Y=y_i)=
\frac{a_{1,0}a_{0,1}}{(a_{0,0}+a_{0,1}+a_{1,0})^2}p(X=0|Y=y_i)^2,$$
or alternatively
\begin{equation}\label{eqn:appendixIV2}a_{0,0}\cdot p(X=1|Y=y_i)=
\frac{a_{1,0}a_{0,1}}{a_{0,0}+a_{0,1}+a_{1,0}}p(X=0|Y=y_i),\end{equation}
If $a_{0,0}$ is zero, either $a_{1,0}$ or $a_{0,1}$ must also be
zero, but this implies that either $U$ or $V$ is a constant random variable which is a contradiction. Therefore $a_{0,0}$ is non-zero.
But then equation \ref{eqn:appendixIV2} uniquely specifies
$p(X=1|Y=y_i)$, implying that $p(X=1|Y=y_1)=p(X=1|Y=y_2)$ which is
again a contradiction.

\section*{Acknowledgement}
The authors would like to thank Chandra Nair for suggesting a
short cut to simplify the original proof, and Young Han Kim for
valuable comments that improved the presentation of the article. The
authors would like to thank TRUST (The Team for Research in
Ubiquitous Secure Technology), which receives support from the
National Science Foundation (NSF award number CCF-0424422) and the
following organizations: Cisco, ESCHER, HP, IBM, Intel, Microsoft,
ORNL, Pirelli, Qualcomm, Sun, Symantec, Telecom Italia and United
Technologies, for their support of this work. The research was also
partially supported by NSF grant numbers CCF-0500234, CCF-0635372,
CNS-0627161 and ARO MURI grant W911NF-08-1-0233 ``Tools for the
Analysis and Design of Complex Multi-Scale Networks."

\end{document}